\documentclass{LMCS}

\def\dOi{10(4:12)2014}
\lmcsheading%
{\dOi}
{1--35}
{}
{}
{Nov.~11, 2013}
{Dec.~18, 2014}
{}

\ACMCCS{[{\bf Mathematics of computing}]: Probability and
  statistics---Stochastic processes;  Probability and
  statistics---Probabilistic representations---Stochastic differential
  equations; [{\bf Theory of computation}]: Models of
  computation---Computability; Randomness, geometry and discrete
  structures---Random walks and Markov chains}


\usepackage{
	amsmath,
	amssymb,
	amsthm,
	color,
	graphicx,
	hyperref,
	stmaryrd 
}
\newcommand{\addtheorem}[2]{
	\newtheorem{#1}[thm]{#2}
}
\addtheorem{con}{Conjecture}
\addtheorem{df}{Definition}
\addtheorem{pro}{Proposition}
\addtheorem{que}{Question}
\addtheorem{ex}{Example}

\newcommand{\la}{\langle}
\newcommand{\ra}{\rangle}

\newcommand{\mc}{\mathcal}

\newcommand{\E}{\mathbb{E}}
\renewcommand{\P}{\mathbb{P}}
\newcommand{\N}{\mathbb{N}}

\newcommand{\openbegin}{(}
\newcommand{\closedbegin}{[}
\newcommand{\openend}{)}
\newcommand{\closedend}{]}
\newcommand{\varW}{X} 
\newcommand{\varvarW}{Y} 
\begin{document}
	\title[Algorithmic randomness for Doob's theorem]{
		Algorithmic randomness for Doob's martingale convergence theorem in continuous time
	}
	\author[B.~Kjos-Hanssen]{Bj\o rn Kjos-Hanssen\rsuper a}
	\address{{\lsuper a}University of Hawai\textquoteleft i at M\=anoa}
	\email{bjoernkh@hawaii.edu}
	\thanks{{\lsuper a}This material is based upon work supported by
the National Science Foundation under Grant No.\ 0901020. This work was
partially supported by a grant from the Simons Foundation (\#315188 to Bj\o
rn Kjos-Hanssen).}

	\author[P.~Nguyen]{Paul Kim Long V.~Nguyen\rsuper b}
	\address{{\lsuper b}University of Hawai\textquoteleft i -- Leeward Community College}
	\email{pvnguyen@hawaii.edu}
        \thanks{{\lsuper b}This material is based upon work supported by the
National Science Foundation under Grant No. 0841223.}

	\author[J.~Rute]{Jason M.~Rute\rsuper c}
	\address{{\lsuper c}Pennsylvania State University}
	\email{jmr71@math.psu.edu}
	\thanks{{\lsuper c}This material is based upon work supported by the
National Science Foundation under Grant No. 0901020.}

	\begin{abstract}
		We study Doob's martingale convergence theorem for computable continuous time
		martingales on Brownian motion, in the context of algorithmic randomness.
		A characterization of the class of sample points for which the theorem holds is given. Such points are given the name of Doob random points.
		It is shown that a point is Doob random if its tail is computably random in a certain sense.
		Moreover, Doob randomness is strictly weaker than computable randomness and is incomparable with Schnorr randomness.
	\end{abstract}
	\maketitle
	\tableofcontents
	\section{Introduction}
		This paper is concerned with computable continuous-time gambling strategies, and their corresponding capital value functions, referred to as martingales.
		First developed by Ville in the 1930's, martingales are a central tool in the study of both algorithmic randomness and classical probability theory.
		
		The classic example of a martingale is as follows.
		A gambler, who has a starting capital of $M_0$, places a wager on the outcome of a fair coin toss.
		If she wins, her capital increases.  If she looses, it decreases by the same amount she would have won.
		This new capital value is denoted $M_1$.
		She continues to place bets on new coin tosses, giving rise to the capital values $M_2, M_3, \ldots$.
		Since the capital value $M_n$ at time $n$ depends on the sequence of coin tosses, we write it as a function $M_n : 2^\N \rightarrow \mathbb{R}$.
		This sequence of capital value functions $(M_n)$ is what is known as a martingale.
		If the martingale is nonnegative (i.e., the gambler never goes into debt),
		then Doob's martingale convergence theorem says that $(M_n)$ converges with probability one as $n\rightarrow \infty$ (Theorem~\ref{thm:doob_discrete}).
		
		Schnorr used martingales to define what is now known as computable randomness.
		Namely, $\alpha \in 2^\mathbb{N}$ is computably random if $\limsup_n M_n(\alpha) <\infty$ for all computable nonnegative martingales $(M_n)$.
		By a folklore result, $\alpha \in 2^\mathbb{N}$ is computably random iff
		$M_n(\alpha)$ converges as $n \rightarrow \infty$ for all computable nonnegative martingales $(M_n)$,
		thereby giving an effective version of Doob's martingale convergence theorem.
		
		Martingale theory is an important tool in modern probability theory and financial mathematics.
		The process of flipping coins can also be thought of as a random walk on the integers (move $+1$ on heads and $-1$ on tails).
		By decreasing both the time step and increment amount to be infinitesimally small, this random walk converges to a continuous process known as Brownian motion.
		A Brownian motion sample path is almost-surely continuous, nondifferentiable, and has high complexity.
		Brownian motion can be used to model many real-world processes.
		Doob developed a rigorous theory of continuous-time martingales ${(M_t)}_{t\in \closedbegin 0,\infty \openend}$,
		including martingales which bet on Brownian motion. 
		(The reader may choose to think of Brownian motion as the value of a stock over time
		and the martingale $(M_t)$ as the value, at time $t$,
		of a stock portfolio which buys and sells that stock according to a certain strategy.)
		Doob's martingale convergence theorem still holds:
		if $(M_t)$ is a nonnegative right-continuous martingale,
		then $(M_t)$ converges with probability one as $t \rightarrow \infty$ (Theorem~\ref{thm:doob_cont}).
		
		The purpose of this paper is to study a new randomness notion on Brownian motion, which we call Doob randomness.
		Namely, a Brownian motion sample path $W\in C \closedbegin 0,\infty \openend$ is
		Doob random (Definition~\ref{df:Doob-random-Brownian}) if $M_t(W)$ converges as $t\rightarrow \infty$ for all computable nonnegative martingales $(M_t)$.
		The reader may be tempted to think of Doob randomness as the continuous-time analogue of computable randomness.
		However, computable randomness already has a more natural definition on Brownian motion \cite{2012arXiv1203.5535R},
		and we show that Doob randomness is a strictly weaker notion than computable randomness.
		This weakness has an intuitive explanation: because our martingales are computable,
		$M_t$ must be finite at all times $t$.
		Hence a martingale cannot gain more than a finite amount of money
		solely based on the non-randomness of some proper \textit{initial segment} of the Brownian motion.
		Indeed, we show there are Doob random paths which are constant in the interval $[0,1]$,
		whereas this property cannot hold for a computably random Brownian motion path.
		For similar reasons, Doob randomness is incomparable with Schnorr randomness (another randomness notion weaker than computable randomness).
		
		Instead, Doob randomness depends on the limiting behavior of the Brownian path as time goes to infinity.
		This means that a path is Doob random if and only if ``the tail of the path is computably random'' in some sense.
		
		We believe Doob randomness is interesting because it adds a new dimension to the study of randomness.
		When a Brownian path fails to be Doob random, it is not just because the path contains nonrandom information,
		but also because there is enough time available to exploit that information to gain capital.
		In other words, the order in which we process information now matters in our randomness test.
		The space of Brownian motion paths is not just a probability space,
		but instead it is a probability space with an implicit ``time structure''.\footnote{%
		This time structure can be made explicit using filtrations.
		See Appendix~\ref{appendix}.
		Another type of probability space with additional structure is a measure preserving system.
		It would be interesting to explore definitions of randomness which take this structure into account.
		}
		
		Our paper is structured as follows. In Section~\ref{sec:background}, we provide background on computable analysis and computable measure theory. In Section \ref{sec:2}, we introduce expectation and  martingales on $2^\N$ and define computable randomness.
		
		Before studying Brownian motion,
		we look at two spaces that also have an interesting implicit notion of ``time'',
		${(2^\N)}^\N$ and $2^{\N\times\N}$.
		We find that these two spaces provide a nice simplified setting for studying Brownian motion.
		
		In Section \ref{sec:3}, for a sequence of sequences $\omega = (\omega_n) \in {(2^\N)}^\N$,
		we consider computable martingales which bet on the entire sequence $\omega_n$ at time $n$.
		Again, the resulting notion of randomness depends on the limiting behavior of $(\omega_n)$ as $n\rightarrow {\infty}$.
		
		In Section \ref{sec:4}, we look at martingales which bet on bits of $2^{\N \times \N}$
		using the lexicographic order of $\N \times \N$,
		thereby introducing limit stages into our martingales.
		This leads to a characterization of Doob randomness via a strategy where
		the gambler is only allowed to bet on finitely many bits of one row\footnote{%
			We say that $(i,j)\in\N\times\N$ belongs to row $i+1$ and column $j+1$, i.e., we think in terms of matrices rather than cartesian planes.
		}, before moving to the next.
		This latter characterization is closely related to a result of Miyabe \cite{MR2675686} generalizing van Lambalgen's theorem to the space ${(2^\N)}^\N$.  We show that his result holds for Schnorr randomness.

		In Section \ref{sec:Brownian}, we use the results of Sections \ref{sec:3} and \ref{sec:4} to characterize Doob randomness on Brownian motion.
		We decompose a Brownian path $W$ on the infinite time interval $0\leq t < \infty$ into countably many independent Brownian paths, each on the time interval $0 \leq t \leq 1$.
		Then we represent each of these by a sequence in $2^\N$, thereby representing our original path $W$ as a sequence in ${(2^\N)}^\N$.
		Using this correspondence, we prove that Doob randomness is strictly weaker than computable randomness and that Doob randomness is incomparable with Schnorr randomness.
		
		This paper brings together a number of fields including algorithmic randomness, computable analysis, measure theory, martingale theory, and Brownian motion.
		In an effort to make this paper more accessible to a broad audience,
		we have tried to limit the use of advanced probability theory.
		The reader does not need to be familiar with stochastic calculus, Brownian motion, or even martingale theory.
		Instead our proofs use basic computable analysis and measure theory, and we carefully spell out the few facts we will need about conditional expectation, martingales, and Brownian motion.
		
		This elementary exposition is most apparent in our definitions of conditional expectation and martingales.
		Instead of giving one general definition of a martingale (as is done in probability texts),
		we give a handful of similar definitions.
		For example, we separately define $\N$-indexed martingales on ${(2^\N)}^\N$,
		$\N \times \N$-indexed martingales on $2^{\mathbb{N}\times\mathbb{N}}$,
		and $\closedbegin 0,\infty \openend$-indexed martingales on the space of Brownian motion.
		This allows us to present elementary proofs of our results, at the cost of a little extra notation.
		These definitions also make it easier for us to reason computably about conditional expectation.
		However, for the reader who wants to know more about the probabilistic background,
		we provide Appendix~\ref{appendix}.
		There we give further background on conditional expectation, filtrations, martingales, and Brownian motion.  We also derive our characterization of conditional expectation from the standard definitions.
	\section{Background}\label{sec:background}
		In this section, we will present relevant background about computable analysis and algorithmic randomness.  		
		We assume the reader is familiar with basic computability theory as well as computability on the real numbers, including computable functions between the spaces $2^\N$, $\N^\N$, and $\mathbb{R}$.
		We also assume the reader has a basic grasp of analysis, including basic measure theory.
		
		The setting of computable analysis is a \emph{computable metric space}, that is a metric space $(M,d)$ with a countable set of points $\{a_n\} \subseteq M$ (called the \emph{simple points}) such that
		$d(a_m,a_n)$ is uniformly computable from the pair $(m,n)$.
		The spaces $2^\N$, $C([0,1])$, and $C(\closedbegin 0,\infty\openend)$ are computable metric spaces under the usual metrics $d(\alpha,\beta)=\inf\{2^{-n}:\alpha_n = \beta_n\}$, $d(f,g) = \|f-g\|_\infty$, and
		$d(f,g)=\sum_n 2^{-n} \min \{\|(f - g) \upharpoonright [0,n]\|_\infty,1\}$ (using any reasonable choice of simple points).
		A \emph{Cauchy name} is an element $f\in\N^\N$ such that for all $n \geq m$ we have $|a_{f(n)} - a_{f(m)}|\leq 2^{-m}$.
		A \emph{computable point} of $M$ is a point with a computable Cauchy name.
		For two computable metric spaces, $M_1,M_2$, a continuous map of type $f:M_1 \rightarrow M_2$ is called \emph{computable} if there is a partial computable map of type $\N^\N \rightarrow \N^\N$ which maps each name of $x$ to a name of $f(x)$.
		We say that $y\in M_2$ is \emph{uniformly computable} from $x \in M_1$ if there is a total computable map $f:M_1 \mapsto M_2$ such that $f(x)=y$.
		(We sometimes drop the ``uniformly'' if it is clear.)
		For example, on the compact space $[0,1]$ (and the same for $2^\N$), the maximum and minimum values of a continuous function $f \in C([0,1])$ are uniformly computable from $f$.
		The computable points of $C([0,1])$ and $C(\closedbegin 0,\infty\openend)$ are exactly the computable functions of type $[0,1] \rightarrow \mathbb{R}$ and $\closedbegin 0,\infty\openend \rightarrow \mathbb{R}$, respectively.
		For more information on these results see the book by Weihrauch \cite{Weihrauch2000}.
		
		There are other computable structures that are not necessarily objects in a computable metric space.
		For example, a \emph{lower semicomputable function} $f:M\rightarrow [0,\infty]$ is the supremum of a computable sequence $(f_n)$ of nonnegative computable functions.
		An \emph{upper semicomputable} function is the infimum.
		We say $r \in [0,\infty]$ is lower semicomputable from $x\in M$ if there is a lower semicomputable function $f:M \rightarrow [0,\infty]$ such that $f(x)=r$.
		An \emph{effectively open} or $\Sigma^0_1$ set is a computable union of balls $B(a,r)$ where $a \in M$ is a simple point and $r \in \mathbb{Q}$.
		Similarly an \emph{effectively closed} or $\Pi^0_1$ set is the complement of a $\Sigma^0_1$ set.
		The computable functions, lower semicomputable functions, upper semicomputable functions, $\Sigma^0_1$ sets, and $\Pi^0_1$ sets are, respectively,
		the computable points in the spaces of continuous functions, lower semicontinuous functions, upper semicontinuous functions, open sets, and closed sets.
		Each of these spaces have a natural representation which assigns to each object in the space a set of \emph{names} in $\N^\N$.
		This also allows us to talk about, say, computable maps from the open sets to the reals.  (See Weihrauch and Grubba \cite{MR2534366} for more details.)
		
		Now we turn to computable measure theory.  More more information see the paper by Hoyrup and Rojas \cite{MR2519075} and the lecture notes by G{\'a}cs \cite{Gacs:fk}.
		\begin{df}\label{df:comp_prob_space}A \emph{computable probability space} $(\Omega, \P)$ is a computable metric space $\Omega$ equipped with a \emph{computable Borel probability measure} $\P$, that is $\P$ is a Borel probability measure on $\Omega$ which satisfies any of the following equivalent properties (see \cite{MR2519075}).
		\begin{enumerate}
		\item{} The map $f \mapsto \int f \,d\P$ is lower semicomputable, where $f$ ranges over nonnegative continuous functions on $\Omega$.
		\item{} The map $f,c \mapsto \int f \,d\P$ is computable, where $f$ ranges over continuous functions  bounded above by $c \in \mathbb{R}$ and bounded below by $-c$.
		\item{} The map $U \mapsto \P(U)$ is lower semicomputable, where $U$ ranges over open sets.
		\item{} The map $C \mapsto \P(C)$ is upper semicomputable, where $C$ ranges over closed sets.
		\end{enumerate}
		\end{df}
		For example, property (3) in the previous definition implies that the measure of $\Sigma^0_1$ sets is lower semicomputable, and the measure of $\Pi^0_1$ sets is upper semicomputable.
		Also for the compact space $2^\N$ (or $[0,1]$), since the maximum and minimum operations are computable,
		$f \mapsto \int f\,d\P$ is a computable map of type $C(2^\N) \rightarrow \mathbb{R}$.
		
		We will have occasion to use both of these randomness notions.  (Computable randomness will be defined in Section \ref{sec:2}, when we introduce martingales.)
		For more information, see Downey and Hirschfelt \cite{MR2732288} for randomness on $2^\mathbb{N}$ and see \cite{MR2519075, Hoyrup-Rojas-Gacs, 2012arXiv1203.5535R} for randomness in a computable metric space.
		
		\begin{df}\label{df:ML_SR}Let $(\Omega, \P)$ be a computable probability space.
			\begin{enumerate}
				\item{} A $\emph{Martin-L{\"o}f test}$ is a computable sequence of $\Sigma^0_1$ sets $U_n$ such that $\P (U_n) \leq 2^{-n}$.
				\item{} A $\emph{Schnorr test}$ is a Martin-L{\"o}f test such that $\P (U_n)$ is uniformly computable from $n$.
				\item{} A point $\omega \in \Omega$ is \emph{Martin-L{\"o}f random} (resp.\ \emph{Schnorr random}) if $\omega \notin \bigcap_n U_n$ for all Martin-L{\"o}f tests (resp.\  Schnorr tests).
			\end{enumerate}	
		\end{df}
	\section{Computably random bit sequences}\label{sec:2}
		In this section, we review computable randomness on $2^\N$ and introduce some notation.
		We also introduce conditional expectation and martingales,
		two concepts that we will use throughout this paper.
		Consider the space $2^\N$ with the fair coin measure $\mathbb{P}$,
		the measure where each bit has an equal likelihood of being $0$ or $1$ independently of the others.
		This is a computable metric space.
		We use the variables $\alpha$, $\beta$ and $\gamma$ for sequences in $2^\mathbb{N}$.
		Let $\alpha_n$ denote the $n+1$st bit of $\alpha$, hence $\alpha = (\alpha_0 ,\alpha_1, \dots)$.
		Then let $\alpha_{<n} = (\alpha_0, \dots, \alpha_{n-1})$,
		$\alpha_{\geq n} = (\alpha_{n} ,\alpha_{n+1}, \dots)$.
		Concatenation of sequences is denoted by ${}^\frown$, so that if $\alpha=(\alpha_0,\dots,\alpha_n)$ is a finite sequence and $\beta = (\beta_0, \beta_1, \dots)$ is a finite or infinite sequence, then
		\[
			\alpha{}^\frown\beta = (\alpha_0,\dots,\alpha_n,\beta_0,\beta_1,\dots).
		\]
		All computable functions $f$ on $2^\N$ are integrable\footnote{%
			Every computable function $f$ is continuous.
			Since $2^{\mathbb N}$ is compact, $f$ is uniformly continuous and bounded, thus integrable.
		}
		and the expectation
		$\E(f) = \int f d\mathbb{P}$ is computable uniformly from $f$ (Definition~\ref{df:ML_SR} and the following paragraph).
		We can define conditional expectation as follows.
		\begin{df}\label{df:En}
			Given an integrable function $f:2^\N \rightarrow \mathbb{R}$ define $\E_n (f)$ as the function given by
			\[
				\E_n (f) (\alpha) = \int\! g(\beta)\, d\mathbb{P}(\beta) = \E(g) \quad \text{where}\quad g(\beta)=f(\alpha_{<n} {}^{\frown}\beta).
			\]
		\end{df}
		Note, if $ f $ is computable, then $\E_n (f) (\alpha)$ is computable uniformly from $f$ and $\alpha_{<n}$.
		\begin{rem}\label{rem:EnProperties}
			Conditional expectation is an important concept in probability theory.
			See Appendix~\ref{appendix} for more information.
			We will introduce a number of different variations of Definition~\ref{df:En} for different spaces and different choices of indices.
			Nonetheless, they will all satisfy the following five important properties
			which are straightforward from the definition of $\mathbb{E}_n$ (compare with Proposition~\ref{prop:cond_exp_properties}).
			Let $f$ and $g$ be integrable functions.
			\begin{enumerate}
				\item{} If $f(\alpha)$ only depends on $\alpha_{<n}$ for all $\alpha$, then $\E_n(f) = f$.
				\item{} $\E_n (cf + g) = c\, \E_n (f) + \E_n (g)$ for $c\in \mathbb{R}$.
				\item{} $|\E_n (f)| \leq  \E_n (|f|)$, and therefore  $\|\E_n (f)\|_{\infty} \leq  \|f\|_{\infty}$.
				\item{} If $f(\alpha)$ only depends on $\alpha_{<n}$ for all $\alpha$, then $\E_n (fg) = f\, \E_n(g)$.
				\item{} If $m \leq n$, then $\E_m (\E_n (f)) = \E_m (f)$.
			\end{enumerate}
		\end{rem}
		\begin{df}\label{df:mart}
			Let $M={(M_n)}_{n\in \mathbb N}$ be a sequence of real-valued integrable functions on ${2^\N}$.
			\begin{itemize}
				\item{} $(M_n)$ is \emph{adapted} if ${M_n(\alpha)}$ depends only on $n$ and ${\alpha_{< n}}$.
				\item{} $(M_n)$ is a \emph{martingale} if
				$(M_n) $ is adapted and
				${\E_m} (M_n) = M_m$ for ${m \leq n}$.
				\item{} $(M_n)$ is \emph{computably adapted} if $M_n(\alpha)$ is uniformly computable from $n$ and $\alpha_{<n}$.
				\item{} $(M_n)$ is a \emph{computable martingale} if it is a \emph{computably adapted} martingale.
			\end{itemize}
		\end{df}
		Here is another characterization of computable martingales.
		\begin{pro}\label{pro:compAdapted}
			A sequence $(M_n)$ is a computable martingale if and only if
			\begin{enumerate}
				\item{} $(M_n)$ is a computable sequence of computable functions, and
				\item{} ${\E_n} (M_{n+1}) = M_{n}$ for all $n$.
			\end{enumerate}
		\end{pro}
		\begin{proof}
			Assume $(M_n)$ is a computable martingale.  Then clearly (1) and (2) hold.
			Conversely, assume that (1) and (2) hold.  Then $M_n$ is a computable function uniformly in $n$.
			Further, $M_n(\alpha) = (\E_n M_{n+1})(\alpha)$ and is therefore computable from $\alpha_{<n}$ and $n$.
			Also, the more general property $m\leq n$, $\E_m M_n = M_m$ follows by (2) and induction.
		\end{proof}
		\begin{rem}
			The more standard notation in computability theory is to write a martingale as a function $d:2^* \rightarrow \mathbb{R}$
			(where $2^*$ is the set of finite binary sequences) such that
			$\frac{1}{2} d(\alpha_{<n} {}^{\frown}0) + \frac{1}{2} d(\alpha_{<n} {}^{\frown}1) = d(\alpha_{<n})$.
			These are the same using $M_n (\alpha) = d(\alpha_{<n})$.
		\end{rem}
		\begin{df}[\cite{MR0414225, MR0354328, MR2732288}]
			A sequence $\alpha \in 2^\N$ is \emph{computably random} if
			$\limsup_n M_n(\alpha) < {\infty}$ for all nonnegative computable martingales $M$.
		\end{df}
		\begin{thm}[{\cite[Theorem 7.1.3]{MR2732288}}]\label{thisLast}
			For $\alpha \in 2^\N$, the following are equivalent:
			\begin{enumerate}
				\item{} $\alpha$ is not computably random.
				\item{} $\liminf_n M_n (\alpha) = {\infty}$ for some computable nonnegative martingale $(M_n)$.
				\item{} $(M_n (\alpha))$ diverges as ${n \rightarrow {\infty}}$ for some computable nonnegative martingale $(M_n)$.
			\end{enumerate}
		\end{thm}
		Theorem \ref{thisLast} is an effective version of Doob's martingale convergence theorem (Theorem~\ref{thm:doob_discrete}),
		which says that nonnegative martingales converge almost surely.
		More information about computable randomness on $2^\mathbb{N}$ is available
		in, for instance, Downey and Hirschfeldt's monograph \cite{MR2732288}.

	\section{Doob random sequences of sequences}\label{sec:3}
		In this section, we consider the space ${(2^\N)}^\N$ with the uniform\footnote{%
			To be clear, $\mathbb P$ is
			the product measure on ${(2^\N)}^\N$ of
			the product measure on $2^{\mathbb N}$ of
			the fair-coin measure on $2=\{0,1\}$.
		} probability measure $\mathbb{P}$.
		We use the variables $\omega$, $\xi$ and $\psi$ for sequences in ${(2^\N)}^\N$.
		Each $\omega\in {(2^\N)}^\N$ is an infinite sequence of infinite sequences $\omega_n$, and we write $\omega = (\omega_0 ,\omega_1, \dots)$.
		Then we let $\omega_{<n} = (\omega_0, \dots, \omega_{n-1})$, $\omega_{\geq n} = (\omega_{n} ,\omega_{n+1}, \dots)$,
		and  let ${}^\frown$ be concatenation of strings.
		The space $({(2^\N)}^\N,\mathbb{P})$ is isomorphic to $2^\N$ with the fair coin measure
		by the computable isomorphism $\omega \mapsto \bigoplus_{n\in\mathbb N} \omega_n$.%
		\footnote{As usual
			$
				{\left(\bigoplus_{n\in\mathbb N}\omega_n\right)}_{\langle n,k\rangle}={(\omega_n)}_k
			$
			where $\langle \cdot,\cdot\rangle$ is a computable bijective pairing function.
		}
		Therefore all computable functions $f$ on ${(2^\N)}^\N$ are integrable and the expectation $\E(f) = \int f d\mathbb{P}$ is computable uniformly from $f$.\footnote{%
		By this we mean that the map $f \mapsto \E(f)$ is a total computable map of type $C({(2^\N)}^\N) \rightarrow  \mathbb{R}$.  See \cite[Corollary 4.3.2]{MR2519075}.
		}
		As before we can define conditional expectation and martingales as follows.
		\begin{df}\label{df:cond_exp_seq_seq}
		The conditional expectation $\E_n$ for an integrable function $f:{(2^\N)}^\N \rightarrow \mathbb{R}$ is
		\[
			\E_n (f) (\omega) = \int\! g(\xi)\, d\mathbb{P}(\xi) = \E(g)\quad\text{where}\quad g(\xi)=f(\omega_{<n} {}^{\frown}\xi).
		\]
		\end{df}
		Again, $\E_n (f) (\omega)$ is computable uniformly from $f$ and $\omega_{<n}$ and the properties of Remark~\ref{rem:EnProperties} still hold.
		\begin{df}\label{df:mart_seq_seq}		
		A \emph{martingale} $(M_n)$ on ${(2^\N)}^\N$ is a sequence of functions such that the following hold almost surely,
		\begin{enumerate}
			\item{} $M_n (\omega)$ depends only on $n$ and $\omega_{<n}$, and
			\item{} $\E_m (M_n) = M_m$ for all $m \leq n$.
		\end{enumerate}
		\end{df}
		\begin{rem}
			Another way to define $\E_n$ and martingales is as follows. Let the projections $X_n:{(2^\N)}^\N \rightarrow 2^\N$ be given by $X_n(\omega) = \omega_n$.
			Let $\mc{F}_n = \sigma(X_0, \dots, X_{n-1})$, i.e., the $\sigma$-algebra generated by the first $n$ projections.
			For a function $f:{(2^\N)}^\N \rightarrow \mathbb{R}$,
			let $\E_n(f)=\E( f \mid \mc{F}_n)$, that is the conditional expectation of $f$ given $\mc{F}_n$.
			This agrees with our definition of $\E_n$ up to a.e.\ equivalence.
			A martingale adapted to $(\mc{F}_n)$ is a sequence $(M_n)$ of integrable functions such that
			$M_n$ is $\mc{F}_n$-measurable and $\E(M_n \mid \mc{F}_m) = M_m$.
			Again, this definition of martingale agrees with ours (up to a.e.\ equivalence).
			See Appendix~\ref{appendix}.
		\end{rem}
		A Schnorr test can be encoded by a function $f\in\mathbb N^{\mathbb N}$ which encodes a listing of basic open sets for each $n$ which union up to $U_n$,
		and also encodes a fast Cauchy sequence of rationals converging to each measure $\mu(U_n)$ \cite{2012arXiv1209.5478M}.
		This allows us to make the following definition.
		\begin{df}[\cite{2012arXiv1209.5478M}]
			A \emph{uniform Schnorr test} is a total computable function
			$\Phi:2^{\mathbb N}\rightarrow\mathbb N^{\mathbb N}$ such that each $\Phi(\gamma)$ encodes a Schnorr test.
			A collection ${(U^\gamma_n)}_{n\in\mathbb N, \gamma\in 2^{\mathbb N}}$ is a \emph{uniform Schnorr test} if
			${(U_n^\gamma)}_{n\in\mathbb N}$ is encoded by $\Phi(\gamma)$.
			Let $\alpha,\beta\in 2^{\mathbb N}$.
			We say that $\alpha$ is \emph{Schnorr random uniformly relative to} $\beta$ if there is no uniform Schnorr test such that $\alpha\in\bigcap_n U_n^\beta$.
			
			Similarly, a collection ${(M^\gamma)}_{\gamma \in 2^{\mathbb N}}$ of martingales on
			$2^{\mathbb N}$ is a \emph{uniform martingale test} if there is some total computable function
			$\Phi:2^{\mathbb N}\rightarrow\mathbb N^{\mathbb N}$ such that each $\Phi(\gamma)$ encodes a
			martingale.
			Let $\alpha,\beta\in 2^{\mathbb N}$.
			We say that $\alpha$ is \emph{computably random uniformly relative to} $\beta$ if there is no
			uniform martingale test such that $\limsup_n M^\gamma_n(\alpha) = \infty$.
			
		\end{df}
		
		\begin{df}\label{df:Doob-rand}
		Let $\omega \in {(2^\N)}^\N$. Say that
			\begin{enumerate}
				\item{} ${\omega}$ is \emph{computably random} (respectively, \emph{Schnorr random})
					if $\bigoplus_n \omega_n$ is.%
				\item{} ${\omega}$ is \emph{e.c.u.\ random (eventually computably uniformly random)}
					if there is some $n$ such that $\omega_{\geq n}$ is computably random uniformly relative to $\omega_{< n}$.%
				\item{} ${\omega}$ is \emph{Doob random}
					if $( M_n (\omega))$ converges as $n \rightarrow {\infty}$ for all nonnegative computable martingales $(M_n)$.
			\end{enumerate}
		\end{df}
		\begin{rem}\label{rem:time}
			Our definitions of computable and Schnorr randomness on ${(2^\N)}^\N$
			agree with the other definitions in the literature.
			In particular, computable and Schnorr randomness are both invariant under computable
			isomorphisms, such as that between $2^\N$ and ${(2^\N)}^\N$ \cite{2012arXiv1203.5535R}.
			This is not true of the definitions of $\E_n$, martingale, e.c.u.\ randomness, and Doob randomness.
			They depend on an implicit notion of time in the space.
			For us, time is expressed as the index $n$ in the conditional expectation $\mathbb{E}_n$.
			(One could also use the probabilistic notion of a filtration to define time.
			See Appendix~\ref{appendix}.)
		\end{rem}

		Fix a standard enumeration of the rationals ${\{q_n\}}_{n\in \N} = \mathbb{Q}$.
		We define a \emph{computable rational-valued function} to be a function $f : {(2^\N)}^\N \rightarrow \mathbb{Q}$ such that
		$f(\omega) = q_{g(\omega)}$ for some computable function $g:{(2^\N)}^\N \rightarrow \N$.
		Since $f$ is a truth-table reduction, $f(\omega)$ is computable from finitely many bits of ${\omega}$.
		In particular, there is some $n$ computable from the index for $f$ such that $f(\omega)$ is uniformly computable from $\omega_{<n}$.
		Hence, the range of $f$ is finite and computable from $f$.\footnote{%
			Note that a computable rational-valued function is not the same as a computable function taking rational values.
			An example of the latter is $f(\alpha)=\inf \{2^{-n} \mid \alpha_{<n} = 0^n\}$.
		}

		\begin{lem}\label{lem:rationalExpectation}
			If $f$ is a computable rational-valued function, then $\E_n(f)$ is a computable rational-valued function uniformly from $f$.
		\end{lem}
		\begin{proof}
			Recall that by definition, $\E_n(f)(\omega)=\E(g)$ for some rational-valued function $g$ computable uniformly from $f$ and ${\omega}$.
			Further, from $g$ we may compute a single number in $\N$ which  codes
			\begin{enumerate}
				\item{} the finitely-many possible rational values $a_1, \dots, a_k$ that $g$ may take,
				\item{} the clopen sets $A_0, \dots, A_k$ such that $A_i = g^{-1}(q_i)$, and
				\item{} the rational values of $\mathbb{P}(A_i)$ for each $i$.
			\end{enumerate}
			Then we may compute the rational value of $\E(g) = \sum_{i=1}^k {a_i\cdot\mathbb{P}(A_i)}$.
		\end{proof}
		It is not too hard to obtain the following analogue of Theorem \ref{thisLast}.
		\begin{thm}\label{thm:DoobSequences}
			For $\omega \in {(2^\N)}^\N$, the following are equivalent:
			\begin{enumerate}
				\item{} ${\omega}$ is not Doob random.
				\item{} $(M_n (\omega))$ diverges as $n \rightarrow {\infty}$ for some computable rational-valued nonnegative martingale $(M_n)$.
				\item{} $\limsup_n M_n (\omega) = {\infty}$ for some computable rational-valued nonnegative martingale $(M_n)$.
				\item{} $\liminf_n M_n (\omega) = {\infty}$ for some computable rational-valued nonnegative martingale $(M_n)$.
			\end{enumerate}
		\end{thm}
		\begin{proof}
			(1) $\Rightarrow$ (2):
				Assume $(M_n (\omega))$ diverges as $n \rightarrow {\infty}$. Either
				\[
					\liminf_n M_n (\omega) = {\infty},\quad\text{or}
				\]
				\[
					{\limsup_n M_n (\omega)} - {\liminf_n M_n (\omega)} > 2^{-k} > 0
				\]
				for some $k\in\mathbb N$.
				Either way there is a $k$ such that it is enough to approximate $(M_n)$ with
				a rational-valued martingale $(L_n)$ such that $\|M_n - L_n\|_{\infty}\leq 2^{-k}$.

				Fix $k$. From the index for each $M_n$, we may effectively find a rational-valued function $N_n$ depending only on $\omega_{<n}$ such that
				$\|N_n - M_n\|_{\infty} \leq 2^{-(n+k+2)}$.
				However, $(N_n)$ may not be a martingale.
				To make it such, let $L_0 = N_0$ and recursively define $(L_n)$ as
				\[
					L_{n+1} = N_{n+1} - \E_n(N_{n+1}) + L_n.
				\]
				This function $L_n$ remains rational-valued by Lemma~\ref{lem:rationalExpectation}.

				We have that $(L_n)$ is a computable martingale as follows.
				For each $n$,
				\begin{align*}
					\E_n(L_{n+1}) &= \E_n\left(N_{n+1} - \E_n(N_{n+1}) + L_n\right) \\
					&= \E_n(N_{n+1}) - \E_n(N_{n+1}) + L_n = L_n.
				\end{align*}

				Furthermore, we show $\|L_n - M_n\|_{\infty}  \leq (2 - 2^{-n}) 2^{-(k+1)} < 2^{-k}$ by induction.
				We have $\|L_0 - M_0 \|_{\infty} \leq 2^{-(k+2)} < 2^{-(k+1)}$ and
				\begin{align*}
					\|L_{n+1} - M_{n+1}\|_{\infty}
					&= \|N_{n+1} - \E_n(N_{n+1}) + L_n - M_{n+1}\|_{\infty}\\
					&\leq \|N_{n+1} - M_{n+1}\|_{\infty}
						+ \|\E_n(N_{n+1}) - M_n\|_{\infty}
						+ \|M_n - L_n\|_{\infty}
				\end{align*}
				By assumption $\|N_{n+1} - M_{n+1}\|_{\infty} \leq 2^{-((n+1)+k+2)}$.
				By the definition of $\E_n$ as well as the definition of martingale, we have
				\begin{gather*}
					\|\E_n(N_{n+1}) - M_n\|_{\infty}
					= \|\E_n(N_{n+1}) - \E_n(M_{n+1})\|_{\infty} \\
					\qquad =  \|\E_n(N_{n+1} - M_{n+1})\|_{\infty}\leq \|N_{n+1} - M_{n+1}\|_{\infty} \leq 2^{-((n+1)+k+2)}
				\end{gather*}
				Hence by the induction hypothesis,
				\begin{align*}
					\|L_{n+1} - M_{n+1}\|_{\infty}
					&\leq 2^{-(n+k+3)}
					+ 2^{-(n+k+3)}
					+ (2-2^{-n})2^{-(k+1)}\\
					&= (2-2^{-(n+1)})2^{-(k+1)}
				\end{align*}
				as desired.

			(2) $\Rightarrow$ (3):
				We apply \emph{Doob's upcrossing method}. Assume $(M_n (\omega))$ diverges as $n \rightarrow {\infty}$, but $\limsup_n M_n (\omega) < {\infty}$.
				There must be two rationals $a$ and $b$ such that $\limsup_n M_n (\omega) > b > a > \liminf_n M_n (\omega)$.
				Define the following times recursively.
				\begin{align*}
					{\sigma}^\text{up}_0 &= 0 \\
					{\sigma}^\text{down}_{k} &= \inf\{n\in {\openbegin} {\sigma}^\text{up}_{k},   {\infty}{\closedend} : M_n \geq b\}\\
					{\sigma}^\text{up}_{k+1} &= \inf\{n\in {\openbegin} {\sigma}^\text{down}_{k}, {\infty}{\closedend} : M_n \leq a\}
				\end{align*}
				We say that the interval $[{\sigma}^\text{up}_k, {\sigma}^\text{down}_k]$ is an \emph{upcrossing} of $[a,b]$,
				and the interval $[{\sigma}^\text{up}_k, {\sigma}^\text{down}_k]$ is a \emph{downcrossing} of $[a,b]$.
				For ${\omega}$, there are infinitely many upcrossings.

				Note that for any $\xi\in {(2^\N)}^\N$, since $(M_n)$ is rational-valued,
				${\sigma}^\text{up}_k(\xi)$ is computable\footnote{That is: if it is finite, we can compute it; if it is infinite, we can determine that it is above each finite number.}
				as an element of the compactification $\N \cup \{{\infty}\}$ of $\N$
				uniformly from $a$, $b$, $k$, and ${(M_m(\xi))}_{m\leq n}$ where $n={\sigma}^\text{up}_k(\xi)$.
				The same holds for ${\sigma}^\text{down}_k(\xi)$.

				We define our new martingale recursively as follows.
				The idea is that $N$ increases on the upcrossings and is constant on the downcrossings. Let $N_0 = M_0$ and
				\[
					N_n =
					\begin{cases}
						({M}_n - {M}_{{\sigma}^\text{up}_k}) + N_{{\sigma}^\text{up}_k}
						& n \in {\openbegin} {\sigma}^\text{up}_k, {\sigma}^\text{down}_k{\closedend}\\
						N_{{\sigma}^\text{down}_k}
						& n \in {\openbegin} {\sigma}^\text{down}_k, {\sigma}^\text{up}_{k+1}{\closedend}
					\end{cases}.
				\]
				It is easy to see that each $N_n$ is rational-valued uniformly in $n$.
				
				Next we shall show that $(N_n)$ is a computable nonnegative martingale.
				Fix $n$.
				We can simulate the cases in the definition by multiplying by indicator functions.
				Let $\mathbf{1}_\text{up} (\xi)$ indicate that
				$n+1 \in {\openbegin} {\sigma}^\text{up}_k(\xi), {\sigma}^\text{down}_k(\xi){\closedend}$ for some $k$
				and let $\mathbf{1}_\text{down}(\xi)$ indicate that ${n+1} \in {\openbegin} {\sigma}^\text{down}_k, {\sigma}^\text{up}_{k+1}{\closedend}$.
				Also, let $k(\xi)$ be the corresponding $k$ for each case.
				Notice, these functions depend only on $\xi_{<n}$.
				(Even though, these functions talk about events at time $n+1$, their values are determined at time $n$.)
				Then by the properties of Remark~\ref{rem:EnProperties}, we have
				\begin{align*}
					\E_n {N_{n+1}}
					&= \E_n \left(\mathbf{1}_\text{up} \cdot \left(({M}_{n+1} - {M}_{{\sigma}^\text{up}_k}) + N_{{\sigma}^\text{up}_k} \right)
					+ \mathbf{1}_\text{down} \cdot N_{{\sigma}^\text{down}_k} \right)\\
					&=  \mathbf{1}_\text{up} \cdot \left( \left(\E_n({M}_{n+1}) - {M}_{{\sigma}^\text{up}_k} \right) + N_{{\sigma}^\text{up}_k} \right)
					+ \mathbf{1}_\text{down} \cdot N_{{\sigma}^\text{down}_k}\\
					&= \mathbf{1}_\text{up} \cdot \left( \left({M}_{n} - {M}_{{\sigma}^\text{up}_k} \right) + N_{{\sigma}^\text{up}_k} \right)
					+ \mathbf{1}_\text{down} \cdot N_{{\sigma}^\text{down}_k}\\
					&= N_n.
				\end{align*}
				The last line follows by using the definition of $N_n$ and considering the following four cases individually:
				$n > {\sigma}^\text{up}_k, n= {\sigma}^\text{up}_k, n > {\sigma}^\text{down}_k, n = {\sigma}^\text{up}_k$.
				Therefore $(N_n)$ is a computable martingale.

				One can see that the martingale is nonnegative by showing (using induction) that on the upcrossing phase, $N_n \geq M_n$.

				Now we show that $N_{{\sigma}^\text{down}_k(\omega)}(\omega) \geq k(b -a)$ for all $k$.
				First of all ${\sigma}^\text{down}_k(\omega)$ is finite for all $k$ since there are infinitely many upcrossings of $M_n(\omega)$.
				Then using the definitions of $(N_n)$, ${\sigma}^\text{down}_k$ and ${\sigma}^\text{up}_k$, it follows by induction that
				\begin{align*}
					N_{{\sigma}^\text{down}_{k+1}(\omega)}(\omega)
					&= \left({M}_{{\sigma}^\text{down}_{k+1}(\omega)} - {M}_{{\sigma}^\text{up}_{k+1}(\omega)} \right) + N_{{\sigma}^\text{up}_{k+1}(\omega)} \\
					&\geq (a-b) + k(a-b) = (k+1)(a+b).
				\end{align*}
				Hence $\limsup_n N_n(\omega) = {\infty}$.

			(3) $\Rightarrow$ (4):
				Assume $\limsup_n M_n (\omega)= {\infty}$.
				We will apply the ``savings property method'': every time our martingale doubles in value, we keep half of the capital and only bet with the remaining half.
				Let $\tau_0 = 0$ and $\tau_{k+1} = \inf\{n > \tau_k: M_n \geq 2M_{\tau_k} \}$.
				For any $\xi \in {(2^\N)}^\N$, since $M_n$ is rational-valued, $\tau_k\in\mathbb N\cup\{{\infty}\}$ is computable from ${(M_m)}_{m \leq n}$ where $n = \tau_k$.

				Define a new martingale $N_n$ by $N_0 = M_0$ and
				\[ N_n =
					\frac12 N_{\tau_k} +
					\frac12 N_{\tau_k} \frac{M_n}{M_{\tau_k}} \quad \text{when } n \in {\openbegin} {\tau}_k, {\tau}_{k+1}{\closedend}.
				\]
				It is easy to see that each $N_n$ is rational-valued uniformly in $n$.
				
				Now we show $(N_n)$ is a computable nonnegative martingale.
				The proof is similar to the ``upcrossing method'' above.
				Fix $n$.
				Let $k(\xi)$ be such that $n+1 \in {\openbegin}\tau_k, \tau_{k+1}{\closedend}$.
				This function depends only on $\xi_{<n}$.
				By the properties of Remark~\ref{rem:EnProperties}, we have
				\begin{align*}
					\E_n {N_{n+1}}
					&= \E_n \left(\frac12 N_{\tau_k} +
					\frac12 N_{\tau_k} \frac{M_n}{M_{\tau_k}} \right)\\
					&= \frac12 N_{\tau_k} +
					\frac12 N_{\tau_k} \frac{\E_n(M_{n+1})}{M_{\tau_k}}\\
					&= \frac12 N_{\tau_k} +
					\frac12 N_{\tau_k} \frac{M_n}{M_{\tau_k}}\\
					&= N_n.
				\end{align*}
				The last line follows by using the definition of $N_n$ and considering the following two cases individually: $n > \tau_k, n= \tau_k$.
				Therefore $(N_n)$ is a computable martingale.

				One can see that the martingale in nonnegative by induction.

				Now we show that $N_{\tau_{k+1}(\omega)}(\omega) \geq \frac{3}{2} N_{\tau_k(\omega)}(\omega) \geq {(\frac{3}{2})}^k M_0 (\omega)$ for all $k$.
				First of all $\tau_k(\omega)$ is finite for all $k$ since $M_n(\omega)$ is unboundedly large.  Then using the definitions of $(N_n)$ and $\tau_k$, we have
				\begin{align*}
					N_{\tau_{k+1} (\omega)}(\omega)
					&= \frac{1}{2} N_{\tau_k (\omega)} +
					\frac{1}{2} N_{\tau_k (\omega)} \frac{M_{\tau_{k+1} (\omega)}}{M_{\tau_k (\omega)}} (\omega) \\
					&\geq \frac{1}{2} N_{\tau_k (\omega)} (\omega) +
					2 \cdot \frac{1}{2} N_{\tau_k (\omega)} (\omega) \\
					&= \frac{3}{2} N_{\tau_k (\omega)} (\omega).
				\end{align*}
				Therefore, $N_{\tau_{k+1} (\omega)}(\omega) \geq {(\frac {3}{2})}^k \cdot M_0 (\omega)$ by induction.
				Hence
				\[
					\liminf_n N_n(\omega) = {\infty}.
				\]
			(4) $\Rightarrow$ (1):
				Immediate.
		\end{proof}
		\subsection{Almost everywhere computable martingales}
			Up till now, it has sufficed to work with total computable martingales.
			However, this will not be the case when we transition to Brownian motion in Section~\ref{sec:Brownian}.
			This subsection will culminate in a characterization of Doob randomness (under the assumption of Schnorr randomness) in terms of a.e.\ computable martingales in Theorem \ref{thm:aeMart} below.
			The proofs in this subsection are general enough that they will also apply to Brownian motion in Theorem~\ref{thm:DoobS} and Lemma~\ref{lem:ae_comp_mart}.
			\begin{df}[{\cite[Definition~5.2]{MR2519075}}, see also \cite{2012arXiv1203.5535R}]\label{df:aeComp}
				An \emph{a.e.\ computable function}  $f:{(2^\N)}^\N \rightarrow \mathbb{R}$ is a computable function $f:A\subseteq {(2^\N)}^\N \rightarrow \mathbb{R}$
				where $A\subseteq{(2^\N)}^\N$ is measure-one and $\Pi^0_2$.  (Equivalently,
				$f$ is a partial computable function given by an partial computable map $g$ which maps $\omega$ to a name for $f(\omega)$.
				In this case the domain of $f$ is the $\Pi^0_2$ domain of $g$.)
			\end{df}
			\begin{df}\label{df:L1comp}
				A function $f:{(2^\N)}^\N \rightarrow \mathbb{R}$ is \emph{$L^1$-computable} if
				there is a computable sequence of bounded computable functions $(g_k)$ such $\|f-g_k\|_{L^1} \leq 2^{-k}$.
			\end{df}
			\begin{df}
				A function $f:{(2^\N)}^\N \rightarrow \mathbb{R}\cup\{+\infty\}$ is
				\emph{a.e.\ lower semi-computable} if $f$ is
				the supremum
				of a computable sequence of a.e.\ computable functions $f_n$,
				with
				\[
				 	\text{domain}(f)=\bigcap_n \text{domain}(f_n).
			 	\]
			\end{df}
			Note that the domain of an a.e.\ lower semi-computable function is the intersection of a uniformly $\Pi^0_2$ sequence of sets and therefore $\Pi^0_2$.
			\begin{pro}\label{pro:aeCompCondExp}
				Assume $f$ is nonnegative and a.e.\ computable.  Then $\mathbb{E}_n f$ is a.e.\ lower semicomputable in $n$ and $f$.
				Moreover, if $\mathbb{E} f$ is computable, then $\mathbb{E}_n f$ is $L^1$-computable.
				Further, if $f$ is bounded by a constant $d$, then $\mathbb{E}_n f$ is a.e.\ computable uniformly in $m$, $d$ and $f$.
			\end{pro}
			\begin{proof}
				First, the expectation of $\mathbb{E} f$ is lower semicomputable in general and it is computable if $f$ is bounded \cite[Proposition~4.3.1, Corollary 4.3.2]{MR2519075}.
				The a.e.\ computable and a.e.\ lower semicomputable results follow from the definition of $\mathbb{E}_n f$ and from the fact that
				$f \mapsto f(\omega_{<n}{}^\frown\xi)$ is a.e.\ computable (as a map from a.e.\ computable functions to a.e.\ computable functions).

				As for $L^1$-computability, if $\mathbb{E} f$ is computable and $f$ is nonnegative, then $f$ is $L^1$-computable \cite{Hoyrup:2009ly}.
				Let $(g_k)$ be the sequence of bounded computable functions from the definition of $L^1$-computability.
				Then $h_k := \mathbb{E}_n g_k$ is computable uniformly from $g_k$ (and its bound).
				Further, we have by Remark~\ref{rem:EnProperties} that $\|\mathbb{E}_n f - h_k\|_{L^1} \leq \|f - g_k \|_{L^1} \leq 2^{-k}$.
			\end{proof}
			We now come to a lemma that will be used on three separate occasions, a very useful generalization of \cite[Lemma 3.11]{MR3123251}.
			\begin{lem}\label{lem:avoidPoints}
				Given (the name of) an a.e.\ computable function $f$ (respectively, (the name of) a sequence of a.e.\ computable functions $(f_i)$),
				one can effectively find a dense sequence of points ${\{a_n\}}_{n\in \mathbb{N}} \subseteq \mathbb{R}$ such that $\mathbb{P}\{\omega \mid f(\omega) = a_n\} = 0$
				(respectively, $\mathbb{P}\{\omega \mid \forall i\  f_i(\omega) = a_n \} = 0$) for all $n$.
				This implies that for any such $a_n$, the measure $\mathbb{P}\{\omega \mid f_i (\omega) \leq a_n\}$ is uniformly computable from $n$ and $i$.
				Moreover, this result also holds for $L^1$-computable functions in place of a.e.\ computable functions.
			\end{lem}
			\begin{proof}
				Consider a computable sequence $(f_i)$ of a.e.\ computable functions (or $L^1$-computable functions).
				Consider the push-forward measures $\mu_i$ given by
				\[
					\mu_i(B)=\mathbb P(\{\omega \mid f_i(\omega)\in B\}).
				\]
				These are computable from $f_i$ \cite{Hoyrup:2009ly}. Fix $k$.
				Now, enumerate all rationals intervals $[q_i,r_i]$ such that
				$\max\{
					\mu_0([q_i,r_i]), \dots, \mu_k([q_i,r_i])
				\}<2^{-k}$.
				This is enumerable, since the measure of a $\Pi^0_1$ set is upper semicomputable \cite[Proposition 4.3.1]{MR2519075}.
				Now let $U_k = \bigcup_i (q_i,r_i)$.  The sequence $(U_k)$ is a uniformly $\Sigma^0_1$ sequence of dense sets.
				Therefore,
				one can compute a dense sequence of points ${a_n}$ in the intersection $\bigcap_k U_k$.%
				\footnote{%
					Starting with any interval $I_0 \subseteq \mathbb{R}$, since each $U_k$ is dense, we can effectively choose a decreasing sequence of intervals
					$I_k = (a_i - \epsilon_k, a_i + \epsilon_k) \subseteq U_k$ such that $I_k \subseteq (a_k - \epsilon_k/2, a_k + \epsilon_k/2)$.  Let $a =\lim_k a_k$.
					In a sense this amounts to the proof of the Baire category theorem being effective; see \cite{MR2519075, 2012arXiv1209.5478M} for similar arguments.
				}
				Notice that any such $a_n$ cannot be an atom of the push-forward measure $\mu_i$,
				and therefore $\mathbb{P}\{\omega \mid f_i(\omega) = a_n\} = 0$.

				To see that $\mathbb{P}\{\omega \mid f_i (\omega) \leq a_n\}$ is computable, just notice that since $a_n$ is not an atom of $\mu_i$ we have
				\[
					\mathbb{P}\{\omega \mid f_i (\omega) \leq a_n\}
					=\mu_i ({\openbegin} {-\infty},a_n{\closedend}) = \mu_i ((-{\infty},a_n)).
				\]
				Since $\mu_i ({\openbegin} {-\infty},a_n{\closedend})$ is the measure of a $\Pi^0_1$ set
				it is upper semicomputable,
				while $\mu_i ({\openbegin} {-\infty},a_n{\openend})$ is the measure of a $\Sigma^0_1$ set
				and is lower semicomputable \cite[Proposition~4.3.1]{MR2519075}.
				Therefore, $\mathbb{P}\{\omega \mid f_i (\omega) \leq a_n\}$ is computable.
			\end{proof}
			\begin{lem}\label{lem:approxAeComp}
				Given $\epsilon > 0$, $\delta > 0$, $n\in \mathbb{N}$, and a nonnegative a.e.~computable function $f$ with a computable expectation, we can find
				a computable function $g$,
				a $\Sigma^0_1$ set $U$ of computable measure, and
				a number $d$ (all uniformly computable from the names for $\varepsilon$, $\delta$, $n$, $f$, and $\mathbb{E} f$) such that
				$\mathbb{P}(U) < \delta$,
				$\| g \|_{\infty} < d$, and
				$|\E_m(f) - \E_m(g)| \leq \varepsilon$ outside of $U$ for all $m \leq n$.
			\end{lem}
			\begin{proof}
				A nonnegative a.e.\ computable function $f$ with a computable expectation is $L^1$-computable \cite{Hoyrup:2009ly}.
				By the definition of $L^1$-computable, we can find a nonnegative bounded computable function $g$ with a bound $d$ such that
				$\|f-g\|_{L^1} \leq \delta \cdot \epsilon/8(n+1)$.
				Unfortunately, since $f$ may not be bounded, $\mathbb{E}_m f$ may not be a.e.\ computable.
				Instead we have that $\mathbb{E}_m f$ is $L^1$-computable and a.e.\ lower semicomputable (Proposition~\ref{pro:aeCompCondExp}).
				Let $h = \min\{f,d\}$.
				Since $h$ and $g$ are bounded by $d$, we have $\mathbb{E}_m g$ and $\mathbb{E}_m h$ are a.e.\ computable.
				Therefore $|\mathbb{E}_m h - \mathbb{E}_m g|$ and $|\mathbb{E}_m f - \mathbb{E}_m h| = \mathbb{E}_m f - \mathbb{E}_m h$
				are a.e.\ lower semicomputable and $L^1$-computable.

				By Remark~\ref{rem:EnProperties}, for all $m$, $\| \mathbb{E}_m h - \mathbb{E}_m g\|_{L^1} \leq \|h-g\|_{L^1} \leq \|f-g\|_{L^1}$ and
				$\| \mathbb{E}_m f - \mathbb{E}_m h\|_{L^1} \leq \|f-h\|_{L^1} \leq \|f-g\|_{L^1}$.
				Let $a \in \{a_n\}$ be from Lemma~\ref{lem:avoidPoints} such that $\varepsilon/4 < a < \varepsilon/2$.
				We have by Markov's inequality that
				\begin{align*}
					\mathbb{P}\underbrace{\{\xi \mid (\mathbb{E}_m f(\xi) - \mathbb{E}_m h(\xi)) > a \}}_{:=V^0_m}
					&\leq \frac{\|f-g\|_{L^1}}{\varepsilon/4}
					\leq \frac{\delta/2}{n+1},\quad\text{and}\\
					\mathbb{P}\underbrace{\{\xi \mid |\mathbb{E}_m h(\xi) - \mathbb{E}_m g(\xi)| > a \}}_{:=V^1_m}
					&\leq \frac{\|f-g\|_{L^1}}{\varepsilon/4}
					\leq \frac{\delta/2}{n+1}.
				\end{align*}
				By Lemma~\ref{lem:avoidPoints}, the set $V^0_m$ has computable measure, but we would also like $V^0_m$ to be $\Sigma^0_1$.
				This depends on which  undefined points of $\mathbb{E}_m f - \mathbb{E}_m h$ we include in $V^0_m$.
				Define $V^0_m$ as the union of all basic open sets $B$ such that the algorithm for $\mathbb{E}_m f - \mathbb{E}_m h$ ensures that
				$(\mathbb{E}_m f- \mathbb{E}_m g)(\xi) > a$ for all $\xi \in B$ such that $\mathbb{E}_m f - \mathbb{E}_m h$ exists.
				Since $(\mathbb{E}_m f- \mathbb{E}_m g)$ is a.e.\ lower semicomputable, by this definition, $V^0_m$ is $\Sigma^0_1$.
				The same can be done with $V^1_m$.

				Finally, let $U = \bigcup_{m\leq n} V^0_m \cup V^1_m$.
				This set is $\Sigma^0_1$ with a computable measure satisfying the desired properties.
			\end{proof}
			\begin{df}\label{df:aeMart}
				A martingale $ (M_n) $ is an \emph{a.e.\ computable martingale} if $M_n (\omega)$ is a.e.\ computable uniformly from $n$ and $\omega_{<n}$.
			\end{df}
			Similar to Proposition~\ref{pro:compAdapted}, a martingale $(M_n)$ is an a.e.\ computable martingale if and only if
			$M_n$ is a.e.\ computable, uniformly in $n$.
			\begin{thm}\label{thm:aeMart}
				Let $\omega \in {(2^\N)}^\N$ be Schnorr random. Then the following are equivalent:
				\begin{enumerate}
				\item{} $(M_n (\omega))$ diverges as $n \rightarrow {\infty}$ for some computable nonnegative martingale $(M_n)$
					with a computable bound on the growth $d:\mathbb{N} \rightarrow \mathbb{R}$ such that $\| M_n \|_{\infty} \leq d(n)$.
				\item{} ${\omega}$ is not Doob random.
				\item{} $(M_n (\omega))$ diverges as $n \rightarrow {\infty}$ for some a.e.\ computable nonnegative martingale $(M_n)$.
				\end{enumerate}
			\end{thm}
			\begin{proof}
				(1) $\Rightarrow$ (2):
					This direction is immediate from the definition of Doob randomness.

				(2) $\Rightarrow$ (3):
					This direction is immediate from the fact that a computable martingale is a.e.\ computable.

				(3) $\Rightarrow$ (1):
					It is enough to approximate the a.e.\ computable $(M_n)$ with a computable martingale $(L_n)$
					such that $|L_n(\omega) - M_n (\omega)|\leq 2^{-k}$ for a large enough $k$.

					By Lemma~\ref{lem:approxAeComp}, from $k$ and from the code for each $M_n$,
					we may effectively find a computable function $N_n$ (depending on $k$) such that
					$|N_n(\xi) - M_n(\xi)| \leq 2^{-(n+k+2)}$ and $|\E_{n-1}(N_n)(\xi) - M_{n-1}(\xi)| \leq 2^{-(n+k+2)}$
					for all $\xi$ outside a $\Sigma^0_1$ set $U^k_n \subseteq {(2^\N)}^\N$ of computable measure at most $2^{-(n+k+1)}$.
					Further, by Lemma~\ref{lem:approxAeComp}, we may assume $\|N_n\|_{\infty} < d(n)$ for a computable bound $d(n)$.
					(Also, the value of $N_n (\xi)$ depends only on $\xi_{<n}$.)
					Then ${\bigcup_n U^k_n}_{n\in \N ,k \in \N}$ is a total Solovay test\footnote{%
					A \emph{total Solovay test} is a sequence (or in this case a doubly-indexed sequence)
					of $\Sigma^0_1$ sets ${U^k_n}$ such that $\sum_n \sum_k \mu(U^k_n)$ is finite and computable.
					By \cite[Theorem 7.1.10]{MR2732288}, every Schnorr random $\omega$ is in only finitely many
					$U^k_n$.
					}
					and therefore there is some large $k$ such that $\omega \notin U^k_n$ for all $n$.

					However, $(N_n)$ may not be a martingale.
					To make it such, let $L_0 = N_0$ and recursively define $L_n$ as
					\[ L_{n+1} = N_{n+1} - \E_n(N_{n+1}) + L_n.\]
					Similarly to the proof of Theorem~\ref{thm:DoobSequences} ((1) $\Rightarrow$ (2)), $(L_n)$ is an a.e.~computable martingale.
					Also, $\|L_n\|_{\infty} \leq \sum_n 2d(n)$.

					To see that  $|L_n(\xi) - M_n(\xi)| \leq 2^{-k}$, we must do calculations similar to that of Theorem~\ref{thm:DoobSequences} ((1) $\Rightarrow$ (2)).
					To show $|L_n(\omega) - M_n(\omega)|  \leq (2 - 2^{-n}) 2^{-(k+1)} < 2^{-k}$ by induction, we will repeat the same calculations as in
					Theorem~\ref{thm:DoobSequences} ((1) $\Rightarrow$ (2)), with the following adjustment.
					Since there is no global bound, we will show the bound directly for ${\omega}$.
					We have $|L_0(\omega) - M_0(\omega) | \leq 2^{-(k+2)} < 2^{-(k+1)}$ and
					\begin{align*}
						|L_{n+1}(\omega) - M_{n+1}(\omega)|
						&= |N_{n+1}(\omega) - \E_n(N_{n+1})(\omega) + L_n - M_{n+1}(\omega)|\\
						&\leq |N_{n+1}(\omega) - M_{n+1}(\omega)|
							+ |\E_n(N_{n+1})(\omega) - M_n(\omega)|\\
							&\qquad + |M_n(\omega) - L_n(\omega)|
					\end{align*}
					By assumption $|N_{n+1}(\omega) - M_{n+1}(\omega)| \leq 2^{-((n+1)+k+2)}$
					and $|\E_n(N_{n+1})(\omega) - M_n(\omega)|\leq 2^{-((n+1)+k+2)}$.
					Hence by the induction hypothesis,
					\begin{align*}
						|L_{n+1}(\omega) - M_{n+1}(\omega)|
						&\leq 2^{-(n+k+3)}
						+ 2^{-(n+k+3)}
						+ (2-2^{-n})2^{-(k+1)}\\
						&= (2-2^{-(n+1)})2^{-(k+1)}
					\end{align*}
					as desired.
			\end{proof}

	\section{Doob random bit arrays}\label{sec:4}
		In this section, consider the space $2^{\N \times \N}$ with the uniform probability measure $\mathbb{P}$.
		We use the variables $\omega$, $\xi$ and $\psi$ for elements of $2^{\mathbb{N} \times\mathbb{N}}$.
		If $\omega \in 2^{\N \times \N}$, let $\omega_{m,n}$ denote the bit in the $m+1$st row and $n+1$st column, hence $\omega = {(\omega_{m,n})}_{(m,n)\in\N\times\N}$.
		Let $<$ denote the lexicographical order on $\N \times \N$; that is, $ (k,\ell) < (m,n) $ if and only if $k<m$ or both $k=m$ and $\ell < n$.
		Then let $\omega_{<(m,n)} = \big(\omega_{k,\ell} \mid (k,\ell)<(m,n)\big)$ and $\omega_{\geq (m,n)} = \big(\omega_{k,\ell} \mid (k,\ell)\geq(m,n)\big)$.
		Since the order type of $\big((k,\ell) \mid (k,\ell)\geq(m,n)\big)$ is the same as $\N \times \N$,
		we will consider $\omega_{\geq (m,n)}$ to be an element of  $2^{\N \times \N}$.
		Concatenation ${}^\frown$ is defined by
		\[
			{(\omega_{<(m,n)} {}^{\frown}\xi)}_{k,\ell} =
			\begin{cases}
				\omega_{k,\ell} & (k,\ell) < (m,n) \\
				\xi_{0,\ell-n} & k=m,\ell \geq n\\
				\xi_{(m-k,\ell)} & k > m
			\end{cases}.
		\]
		In particular, $\omega_{<(m,n)} {}^{\frown}\omega_{\geq (m,n)} = \omega$.

		We will naturally identify the spaces ${(2^\N)}^\N$ and $2^{\N\times\N}$
		by denoting the $m+1$st sequence of ${\omega}$ as the row $\omega_m=(\omega_{m,0},\, \omega_{m,1},\, \dots)$.
		Then we may associate $\omega_{< m}$ with $\omega_{< (m,0)}$ and $\omega_{\geq m}$ with $\omega_{\geq (m,0)}$.
		The definitions of conditional expectations $\E_n$ and martingales ${(M_n)}_{n\in \N}$ carry over from Definitions~\ref{df:cond_exp_seq_seq} and \ref{df:mart_seq_seq}, except that we now refer to such martingales as \emph{$\N$-indexed martingales}.
		The definitions of computable, ecu, Doob, and Schnorr randomness also carry over.
		\begin{rem}
			We are identifying ${(2^\N)}^\N$ with $2^{\N\times\N}$, but not with $2^\N$.
			To explain this, note that any identification of $2^{\N\times\N}$ and $2^\N$ will not preserve the notion of ``time'' (see Remark~\ref{rem:time}).
			On the other hand,  a subset of the $\N\times\N$-indices (times) correspond to the $\N$-indices (times)---identify $(m,0)$ with $m$.
			Therefore, it easy to view $2^{\N\times\N}$ as a refinement of ${(2^\N)}^\N$ which preserves the ``time structure'', except adding additional times.
			(Again, this can be made formal with filtrations. See Appendix~\ref{appendix}.)
		\end{rem}
		\begin{df}\label{df:cond_exp_mart_array}
			We define conditional expectation for indices in $\N\times\N$ by
			\[
				\E_{m,n} (f) (\omega)= \int g(\xi)\, d\mathbb{P}(\xi) = \E(g) \quad \text{where}\quad g(\xi) = f(\omega_{<(m,n)} {}^{\frown}\xi).
			\]
			Similarly, define an \emph{$\N\times\N$-indexed martingale} ${(M_{m,n})}_{(m,n)\in \N \times \N}$ as a sequence of functions almost surely satisfying
			\begin{enumerate}
				\item{} $M_{m,n} (\omega)$ depends only on $(m,n)$ and $\omega_{<(m,n)}$, and
				\item{} $\E_{k,\ell} (M_{m,n}) = M_{k,\ell}$ for $(k,\ell) \leq (m,n)$.
			\end{enumerate}
		\end{df}
		All the conditional expectation properties of Remark~\ref{rem:EnProperties} still hold.
		This next proposition shows that these $\N\times\N$-indexed conditional expectations and martingales are just extensions of the $\N$-indexed ones.
		\begin{pro}\label{pro:mart_embedding}
			The operators $\E_{m,0}$ and $\E_m$ are equal, and if $(M_{m,n})$ is an $\N\times\N$-indexed martingale, then $N_m = M_{m,0}$ is an $\N$-indexed martingale.
		\end{pro}
		\begin{proof}
			To see that $\E_{m,0} =\E_m$, take an integrable function $f$. Then by the identification of $2^{\N \times\N}$ and ${(2^\N)}^\N$,
			\begin{align*}
				\E_{m,0} (f) (\omega) &= \int_{2^{\N\times\N}} f (\omega_{< (m,0)} {}^{\frown} \xi)\, d\mathbb{P} (\xi)\\
				&= \int_{{ (2^\N)}^\N} f (\omega_{<m} {}^{\frown} \psi)\, d\mathbb{P} (\psi) = \E_{m} (f) (\omega).
			\end{align*}

			As for $N_m = M_{m,0}$, first notice that $M_{m,0}$ is determined by $\omega_{<m,0}=\omega_{<m}$.
			Then for $k < m$, $\E_{k} N_m = \E_{k,0} M_{m,0} = M_{k,0} = M_k$.  So $(N_m)$ is a martingale.
		\end{proof}

		\subsection{Doob randomness via martingales}
			In this subsection, we characterize Doob randomness via $\N\times\N$-indexed martingales.

			Say that $M_{n,m} (\omega)$ \emph{converges to} $a$ as $(m,n)\rightarrow {\infty}$ (or $\lim_{m,n} M_{m,n} (\omega) = a$) if
			for all $\varepsilon>0$ there is a pair $(k,\ell)$ such that for all $(m,n) > (k,\ell)$, $ |M_{m,n} - a| < \varepsilon $.
			The notions of $\limsup$ and $\liminf$ are similar.
			\begin{thm}\label{thm:DoobArray}
				For $\omega \in 2^{\N \times \N}$, the following are equivalent:
				\begin{enumerate}
					\item{} ${\omega}$ is not Doob random.
					\item{} $(M_{m,n} (\omega))$ diverges as $(m,n) \rightarrow {\infty}$ for some computable nonnegative martingale $(M_{m,n})$.
					\item{} $(M_{m,n} (\omega))$ diverges as $(m,n) \rightarrow {\infty}$ for some computable rational-valued nonnegative martingale $(M_{m,n})$.
					\item{} $\liminf_{m,n} M_{m,n} (\omega) = {\infty}$ for
					some computable rational-valued nonnegative martingale $(M_{m,n})$ and
					some computable function $k : \N \rightarrow \N$ such that
					$M_{m,n} = M_{m+1,0}$ for all $n \geq k(m)$.
				\end{enumerate}
			\end{thm}
			\begin{proof}
				(1) $\Rightarrow$ (2):
					Take a computable martingale ${(M_m)}_{m\in \N}$ and extend it to a computable martingale
					${(N_{m,n})}_{(m,n)\in\N\times\N}$ by  $N_{m,n} = \E_{m,n}(M_{m+1})$.
					This is a martingale since $\E_{m,n}(M_{m+1})$ depends only on $\omega_{<m,n}$ and since for $(k,\ell) \leq (m,n)$,
					\begin{align*}
						\E_{k,\ell} N_{m,n}
						&= \E_{k,\ell} (\E_{m,n}(M_{m+1}))
						= \E_{k,\ell} (M_{m+1})
						= \E_{k,\ell} (\E_{k + 1, 0}  (M_{m+1})) \\
						&= \E_{k,\ell} (\E_{k + 1}  (M_{m+1}))
						= \E_{k, \ell}  (M_{k+1}) = N_{k,\ell}.
					\end{align*}

				(2) $\Rightarrow$ (3):
					First consider the $\N$-indexed martingale $N_m = M_{m,0}$ (Proposition~\ref{pro:mart_embedding}).
					By the proof of Theorem~\ref{thm:DoobSequences} ((1) $\Rightarrow$ (2)),
					we may adjust $(N_m)$ to be a rational-valued computable martingale $(L_m)$ such that $\|L_m - N_m\|_{\infty} < 2^{-m}$.
					As in (1) $\Rightarrow$ (2), we can extend $(L_m)$ to an $\N \times \N$-indexed martingale $(K_{m,n})$.  Last we have that
					\begin{align*}
						\|K_{m,n} - L_{m,n}\|
						&= |\E_{m,n}(K_{m+1,0}) - \E_{m,n}(M_{m+1,0})| \\
						&= |\E_{m,n}(L_{m+1}) - \E_{m,n}(N_{m+1})| \\
						&= | \E_{m,n}(L_{m+1} - N_{m+1}) |\\
						&\leq \E_{m,n}(|L_{m+1} - N_{m+1}|) \leq 2^{-(m+1)}.
					\end{align*}

					Hence $(K_{n,m}(\omega))$ diverges since $(M_{n,m}(\omega))$ does.

				(3) $\Rightarrow$ (4):
					Let $(M_{m,n})$ be as in (3).  Since each $M_{m+1,0}(\xi)$ is a computable rational-valued function, it can only depend on finitely many bits of
					$\xi_{<(m+1,0)}=\xi_{<m+1}$.
					Hence there is some $k=k(m,\xi)$ such that $M_{m+1,0}(\xi)$ only depends on $\xi_{<(m,k)}$.
					Since $\xi \mapsto k(m,\xi)$ is a computable function uniformly in $m$ on the effectively compact space $2^{\N \times \N}$,
					it has a maximum which is uniformly computable from $m$.
					So we may assume $k=k(m)$ is that maximum.
					Since $M_{m+1,0}$ only depends on the bits with index below $(m,k)$ we have
					$M_{m+1,0} = \E_{m,k} (M_{m+1,0}) = M_{m,k}$, and the same for all $n\geq k$.

					We can modify $(M_{m,n})$ to a martingale $(L_{m,n})$ with the same properties, but also such that $\liminf L_{m,n} (\omega) ={\infty}$.
					To do this, we can just use the upcrossing and savings techniques in the proofs in Theorem~\ref{thm:DoobSequences}.
					The proofs are the same, except with the following changes to the indices.   Use the indices $(m,n)$ such that $n < k(m)$.
					Since the order type of this set of indices is that of $\N$, the proofs are the same.

				(4) $\Rightarrow$ (1):
					Since $\liminf_{m,n} M_{m,n} (\omega) = {\infty}$, so does the martingale $(N_m)$ given by $N_m = M_{m,0}$.
			\end{proof}
			\begin{rem}\label{rem:gameArray}
				Using (4) of Theorem~\ref{thm:DoobArray} we can characterize Doob randomness in terms of a gambling strategy on bits.
				The gambler starts by betting on the bits of $\omega_{0,0},\,\omega_{0,1},\, \dots$.
				At some stage of the gambler's choosing, she must progress to the next row and start betting on the bits of $\omega_{1,0},\,\omega_{1,1},\, \dots$ in order.
				However, now the gambler may use all the bits of $\omega_0=(\omega_{0,0},\,\omega_{0,1},\, \dots)$,
				including the bits she never bet on, as an oracle for future betting.
				This process continues, with the condition that
				the gambler must eventually progress to each row $\omega_n$ (where she can then use $\omega_{<n}$ as an oracle).
				Her strategy must be total, in that she must specify for all possible ${\omega}$ how to bet and when to progress from $\omega_n$ to $\omega_{n+1}$ for all $n$.
				While the choice of when to progress to the next row may be adaptive (depending on the bits seen so far),
				Theorem \ref{thm:DoobArray} shows that it is sufficient for it to be independent of the array being bet on.
			\end{rem}
		\subsection{Doob randomness via computable randomness}
			We can also characterize Doob randomness entirely in terms of computable randomness.
			\begin{df}
				Let $C$ be an infinite, computable subset of $\N\times\N$.
				For $\omega \in 2^{\N\times\N}$, define $\omega_C \in 2^\N$ to be the bits of ${\omega}$ with indices in $C$ listed in some fixed computable ordering of $C$.
				Given a computable function $f:\N \rightarrow \N$, let
				\begin{align*}
					A_f = \{(m,n)\mid f(m) > n\},&\quad\text{and} \\
					B_f = \{(m,n)\mid f(m) \leq n\}.&
				\end{align*}
				($A$ for ``above'', $B$ for ``below''.)
				This definition will also be extended to noncomputable functions $f$.  In this case, $A_f$ and $B_f$ are ordered by an ordering computable from $f$.
			\end{df}
			\begin{lem}\label{lem:preservationOfCR}
				Computable randomness is preserved by computable permutations.
				Moreover, assume $\alpha \in 2^\N$ is computably random uniformly relative to $\beta \in 2^\N$.
				Then any computable permutation of the bits in $\alpha$ is still computably random uniformly relative to $\beta$.
			\end{lem}
			\begin{proof}
				See \cite[Introduction]{MR2956011} and \cite{MR1769372} for the result that computable randomness is preserved by permutations.
				The proof indicated there can be uniformly relativized to an oracle $\beta$.
			\end{proof}
			\begin{thm}\label{thm:DoobCharacterization}
				Each $\omega \in 2^{\N \times \N}$ is Doob random if and only if
				for all computable functions $f:\N \rightarrow \N$, ${\omega_{B_f}}$ is computably random uniformly relative to $\omega_{A_f}$.
				(By Lemma~\ref{lem:preservationOfCR}, the choice of ordering for $A_f$ and $B_f$ does not matter.)
			\end{thm}
			\begin{rem}
				This can again be characterized as a game similar to Remark~\ref{rem:gameArray}.  Now the gambler may bet on any bit he chooses.
				However, as before, for each row he must decide when to stop betting on that row.
				Yet, as before, he may still use the other bits in that row as an oracle for future bets.
			\end{rem}
			For the proof of Theorem~\ref{thm:DoobCharacterization},
			we will need yet another generalization of conditional expectation and martingales, this time indexed by subsets of $\N\times\N$.
			\begin{df}\label{df:cond_exp_set_index}
				Given a computable set $D \subseteq \mathbb{N}\times\mathbb{N}$ with complement $D^c$,
				define the concatenation operator ${}^{\frown_D}$ to be the unique computable operator
				$2^\mathbb{N} \times 2^\mathbb{N} \rightarrow 2^{\N \times \N}$ such that
				for all $\alpha,\beta \in 2^\mathbb{N}$, if $\omega = \alpha\, {}^{\frown_D}\, \beta$,
				then $\omega_D = \alpha$ and $\omega_{D^c} = \beta$.
				For an integrable function $f:2^{\N\times\N}\rightarrow \mathbb{R}$, define
				\[
					\E_D (f)(\omega) = \int\! g(\xi) \,d\mathbb{P}(\xi) = \E(g)
					\quad\text{where}\quad g(\xi)=f(\omega_D\, {}^{\frown_D}\, \xi).
				\]
			\end{df}		
			
			Notice that $\E_{m,n} = \E_D$ where $D=\{(k,\ell) \mid (k,\ell) < (m,n)\}$.
			We will abbreviate this set as $\{{<}(m,n)\}$.
			The corresponding properties of Remark~\ref{rem:EnProperties} still hold---for example, if $C \subseteq D$, then $\E_C (\E_D f)=\E_C (f)$.
			For a sequence $C_0 \subseteq C_1 \subseteq \dots$ of computable sets, define a martingale $(M_{C_i})$ as a sequence of integrable functions such that
			(1) $M_{C_i}$ only depends on $C_i$ and $\omega_{C_i}$ and (2) for all $i<j$, $\E_{C_i}(M_{C_j}) = M_{C_i}$.

			\begin{proof}[Proof of Theorem~\ref{thm:DoobCharacterization}]
				($\Leftarrow$): Assume ${\omega}$ is not Doob random.  Let $(M_{m,n})$ and $k:\N\rightarrow\N$ be as is (3) of Theorem~\ref{thm:DoobArray}.
				As mentioned in Remark~\ref{rem:gameArray}, $(M_{m,n})$ bets on the bits of $\omega_{B_k}$ using $\omega_{A_k}$ (uniformly) as an oracle.
				A little thought reveals that the order that $(M_{n,m})$ bets on the bits of $\omega_{B_k}$ is of order-type $\N$.
				(That is, there are no limit stages.)
				Therefore, we can use $(M_{m,n})$ to construct an $\N$-indexed martingale which bets on
				$\omega_{B_k}$ (in some order) using $\omega_{A_k}$ uniformly as an oracle.
				Thereby, $\omega_{B_k}$ is not computably random uniformly relative to $\omega_{A_k}$.
				(Lemma~\ref{lem:preservationOfCR} shows that the choice of ordering of $A_k$ and $B_k$ is not important.)

				($\Rightarrow$): Assume ${\omega_{B_f}}$ is not computably random uniformly relative to $\omega_{A_f}$.
				By Lemma~\ref{lem:preservationOfCR} we may assume that $B_f$ is ordered using the lexicographical ordering of $\N \times \N$.
				Then there is a uniformly computable rational-valued martingale $(N^\beta_n)$
				(i.e.\ for all $\beta\in 2^\N$, $(N^\beta_n)$ is a martingale on $2^\N$ uniformly computable from $\beta$)
				such that $\liminf_n N^{\omega_{A_f}}_n (\omega_{B_f}) = {\infty}$.
				We may assume for every $\beta$ that $(N^\beta_n)$ has a savings property such that $N^\beta_n \geq N^\beta_m/2$ for all $m\leq n$.

				Our goal is to convert this class of martingales $(N^\beta_n)$ into a martingale $(K_{m,n})$ on $2^{\N \times \N}$.
				This is complicated by the fact that
				our uniform martingale $(N^\beta_n)$ looks at the bits in a different order than the lexicographical ordering on $\N \times \N$.
				The advantage of martingales indexed by sets is that it gives us a language for talking about betting on bits in a different order.

				We will convert the class of martingales $(N^\beta_n)$ as a single martingale indexed by sets.
				Let $B=B_f$ and $A=A_f$.  Let $A=\{a_n\}$ and $B=\{b_n\}$ be computable enumerations of $A$ and $B$.
				Let $A_n = \{a_0, \dots, a_{n-1}\}$ and $B_n = \{b_0, \dots, b_{n-1}\}$.
				Then we define the martingale  $M_{A \cup B_n} (\xi) = N^{\xi_A}_n (\xi_B)$.
				We will first show that $(M_n)$ is a computable martingale. Notice that $(A \cup B_0)\subseteq (A \cup B_1) \subseteq \dots$.
				Then $M_{A\cup B_n} (\xi)=N^{\xi_A}_n (\xi_B)$ is uniformly computable from $\xi_{A}$ and ${(\xi_B)}_{< n}$,
				or equivalently, $\xi_{A \cup B_n}$.  Lastly, for $m \leq n$, we can unwrap the definitions to get
				\begin{align*}
					\E_{A\cup B_m} (M_{A \cup B_n})(\xi)
					&= \int_{2^{\N\times\N}} \! M_{A \cup B_n} (\xi_{A \cup B_m} {}^{\frown_{A \cup B_m}} \psi)\, d\psi\\
					&= \int_{2^{\N\times\N}} \! N^{\xi_A}_n ( {(\xi_{A \cup B_m} {}^{\frown_{A \cup B_m}} \psi)}_B)\, d\psi\\
					&= \int_{2^\N} {\!} N^{\xi_A}_n [ {(\xi_B)}_{<m} {}^{\frown}{\gamma}]\, d\gamma\\
					&= \E_m (N^{\xi_A}_n) (\xi_B)
					= N^{\xi_A}_m (\xi_B)
					= M_{A\cup B_m} (\xi).
				\end{align*}
				Hence $(M_{A\cup B_n})$ is a computable martingale.

				Since $M_{A\cup B_n}$ is rational-valued,
				it only depends on finitely many bits of $\xi$.
				In fact, $M_{A\cup B_n}$ is truth-table computable in the following sense:
				There is some computable function $\ell:\N \rightarrow \N$ such that
				$M_{A\cup B_n}(\xi)$ is computable uniformly from $n$, $\xi_{A_{\ell(n)}}$, and $\xi_{B_n}$.

				Recall our assumption that $B={b_n}$ is ordered lexicographically.
				Define $C_n = \{(m,k) \mid (m,k) < b_n\} \in \N \times \N$.
				Then $A_{\ell{n}} \cup C_n$ are the infinitely many positions $C_m$ that come before $b_n$
				as well as the finitely many positions in $A_{\ell{n}} \smallsetminus C_n$ which represent oracle bits looked at in the future.
				The goal is get rid of the dependence on $A_{\ell{n}} \smallsetminus C_n$.
				Notice $A_{\ell(m)}\cup C_m \subseteq A\cup B_m$.

				Let $L_{A_{\ell(n)} \cup C_n} (\xi) =  M_{A \cup B_n} (\xi)$. Now we can show that $(L_{A_{\ell(n)} \cup C_n})$ is also a martingale.
				For $m < n$ we have
				\begin{align*}
					&\E_{A_{\ell(m)} \cup C_m}  L_{A_{\ell(n)} \cup C_n} \\
					&\quad =  \E_{A_{\ell(n)} \cup C_m} \left(\E_{A\cup B_m} M_{A \cup B_n} \right)
					&\quad
					\text{
						($A_{\ell(m)}\cup C_m \subseteq A\cup B_m$)
					}\\
					&\quad = \E_{A_{\ell(m)}\cup C_m} M_{A\cup B_m} \\
					&\quad = M_{A\cup B_m} = L_{A_{\ell(m)} \cup C_m} .
					& (M_{A\cup B_m}\text{ only depends on} \\
					&&\text{the bits in }A_{\ell(m)}\cup C_m)
				\end{align*}

				Last, we define a martingale on $2^{\N \times \N}$ by
				\[ K_{m,n} = \E_{m,n} L_{A_{\ell(k)}\cup C_k} \quad\text{where}\quad k=\min\{i :  (m,n)\leq b_i \in B\}. \]
				The following calculation shows that $K_{m,n}$ is a martingale.  Let $K_{m,n}$ be as in the previous formula.
				Assume $(m',n') \leq (m,n)$ and $k'\leq k$ such that $k=\min\{i :  (m,n)\leq b_i \in B\}$ and $k'=\min\{i :  (m',n')\leq b_i \in B\}$.  Then
				\begin{align*}
					\E_{m',n'}  K_{m,n}
					&=  \E_{m',n'} \left(\E_{m,n} L_{A_{\ell(k)}\cup C_k} \right) \\
					&= \E_{m',n'} L_{A_{\ell(k)}\cup C_k} \\
					&= \E_{m',n'} \left(\E_{A_{\ell(k')}\cup C_{k'}} L_{A_{\ell(k)}\cup C_k} \right) \\
					&= \E_{m',n'} L_{A_{\ell(k')}\cup C_{k'}} = K_{m',n'}.
				\end{align*}

				Finally, we show that $\limsup_{m,n} L_{m,n}(\omega) = {\infty}$.  Pick a large number $c>0$.
				Since $\liminf_k N^{\omega_A}_k(\omega_B) = {\infty}$, there is an $k_0$ such that $N_{k_0}(\omega) \geq 2 c$.
				Let $k > k_0$. Then by the above definitions as well as the savings property,
				\begin{align*}
				K_{b_k}(\omega)
				&= \E_{b_k} L_{A_{\ell(k)} \cup C_k} (\omega)
				= \E_{b_k} M_{A \cup B_k} (\omega)
				= \E_{b_k} N^{\omega_A}_k (\omega_B) \\
				&\geq \E_{b_k} \frac{N^{\omega_A}_{k_0} (\omega_B)}{2} \geq \E_{b_k} c = c.
				\end{align*}
				Hence $\limsup_{m,n} L_{m,n}(\omega) = {\infty}$ as desired and ${\omega}$ is not Doob random.
			\end{proof}
		\subsection{Doob, Schnorr, and computable randomness}
			Here we show that Doob randomness is strictly weaker than computable randomness and incomparable with Schnorr randomness.
			\begin{thm} [Miyabe \cite{Miyabe:2011zr, 2012arXiv1209.5478M}]\label{thm:vLcr}
				For $\alpha,\beta,\gamma\in2^\N$, if $\alpha\oplus\beta$ is computably random (uniformly relative to $\gamma$)
				then $\alpha$ is computably random uniformly relative to $\beta$ (respectively, $\beta\oplus\gamma$).
			\end{thm}
			\begin{proof} This is the same proof as in \cite{Miyabe:2011zr} uniformly relativized to $\gamma$.
			\end{proof}
			\begin{cor}\label{cor:vLcr}
			Let $\omega \in 2^{\N \times \N}$ and let $D \subset \N \times \N$ be an infinite computable set.
			If ${\omega}$ is computably random (uniformly relative to $\gamma$)
				then $\omega_D$ is computably random uniformly relative to $\omega_{D^c}$ (respectively, uniformly relative to $\omega_{D^c}$ and $\gamma$).
			\end{cor}
			\begin{proof} Follows from Theorem~\ref{thm:vLcr} and Lemma~\ref{lem:preservationOfCR}.
			\end{proof}
			\begin{thm}\label{thm:crDoob}
			For $\omega \in 2^{\N \times \N}$, consider the the following:
				\begin{enumerate}
				\item{} ${\omega}$ is computably random.
					\item{} ${\omega}$ is e.c.u.\ random.
					\item{} ${\omega}$ is Doob random.
				\end{enumerate}
			We have (1) $\Rightarrow$ (2) $\Rightarrow$ (3) and both implications are strict.
			\end{thm}
			\begin{proof}
				(1) $\Rightarrow$ (2):
					This follows from Corollary~\ref{cor:vLcr} and the definition of e.c.u.\ random.

				(2) $\Rightarrow$ (3):
					Assume that $\omega_{\geq n}$ is computably random uniformly relative to $\omega_{<n}$ for some $n$.
					Fix a computable function $f:\N\rightarrow\N$.
					By Theorem~\ref{thm:DoobCharacterization}, it suffices to show that ${\omega_{B_f}}$ is computably random uniformly relative to $\omega_{A_f}$.
					We may assume that $f(k)=0$ for each $k < n$,
					since it would only add and remove finitely many bits to/from $\omega_{A_f}$ and ${\omega_{B_f}}$.
					Let $C = \{(k,\ell) \in A_f \mid k \geq n\}$ be the coordinates of $A_f$ for which the first coordinate is $\geq n$,
					and let $\omega_C$ be the corresponding bits of ${\omega}$ listed in a computable order.
					Note that $\omega_{A} \oplus \omega_{B}$ is a computable permutation of the bits in $\omega_{\geq n}$
					and that computable randomness is preserved under computable permutations (Lemma~\ref{lem:preservationOfCR}).
					Then $\omega_{A} \oplus \omega_{B}$ is computably random uniformly relative to $\omega_{<n}$.
					Hence by Theorem~\ref{thm:vLcr}, $\omega_B$ is computable random uniformly relative to
					$\omega_C \oplus \omega_{<n}$ which is just a computable permutation of $\omega_{A_f}$.

				(2) $\not\Rightarrow$ (1):
					Let ${\omega}$ be such that the first row $\omega_0 = (0,0,\dots,)$ but $\omega_{\geq 1}$ is computably random.
					Clearly ${\omega}$ is e.c.u.\ random, but not computably random.

				(3) $\not\Rightarrow$ (2):
					Let $g$ be a function which dominates all computable functions $f$, that is $g(n) > f(n)$ for all but finitely many $n$.
					Define ${\omega}$ as follows.
					Choose $\alpha \in 2^\N$ to be Martin-L{\"o}f random relative to $g$.%
					\footnote{%
						We could let $\alpha$ be computably random uniformly relative to $g$ instead.
						However, computable randomness uniformly relative to an oracle in $\N^\N$ is not in the existing literature.
						Therefore, we used a stronger notion of randomness, (defined in, e.g., \cite[Section 6.4]{MR2732288}).
					}
					Let $B_g$ be the set of coordinates below $g$.  Give $B_g$ an ordering computable from $g$.
					Then let ${\omega_{B_g}}$ be the bits of $\alpha$ put into the positions in $B_g$ using the ordering on $B_g$.
					Let ${\omega_{A_g}}$, i.e., $\omega$ in the positions above $g$, be all $0$s.  Then each row $\omega_n$ ends in $0$s.
					First, we show that ${\omega}$ is not e.c.u.\ random.
					By Corollary~\ref{cor:vLcr}, if $\omega_{\geq n}$ is computably random, then $\omega_n$ would be computably random which it is not.
					So ${\omega}$ is not e.c.u.\ random.

					Now, we show that ${\omega}$ is Doob random using the characterization in Theorem~\ref{thm:DoobCharacterization}.
					The proof is similar to (2) $\Rightarrow$ (3). Choose a computable function $f$.
					Let $A_f$ and $B_f$, respectively, be the coordinates above and below $f$.  Let $C_{f,g}$ be the coordinates above $f$ and below $g$.
					Without loss of generality, we may assume $f < g$ for all values,
					for this only changes finitely many bits of $\omega_{A_f}$ and ${\omega_{B_f}}$.

					We want to show that ${\omega_{B_f}}$ is computably random uniformly relative to $\omega_{A_f}$.

					Recall that $\omega_{B_g}$ is Martin-L{\"o}f random relative to $g$.
					By van Lambalgen's theorem for Martin-L{\"o}f randomness relativized to $g$
					(similar in form to Theorem~\ref{thm:vLcr}, see \cite[Theorem~6.9.1]{MR2732288}),
					we have that ${\omega_{B_f}}$ is Martin-L{\"o}f random relative to $\omega_{C_{f,g}}$ and $g$.
					Hence ${\omega_{B_f}}$ is also Martin-L{\"o}f random relative to $\omega_{A_f}$, since
					$\omega_{A_f}$ is computable from $\omega_{C_{f,g}}$ and $g$.
					Then, ${\omega_{B_f}}$ is also computably random uniformly relative to ${\omega_{A_f}}$  (see \cite{Miyabe:2011zr, 2012arXiv1209.5478M}).
			\end{proof}
			\begin{thm}\label{thm:SchnorrDoob}
			For $\omega \in 2^{\N \times \N}$, consider the following:
				\begin{enumerate}
				\item{} ${\omega}$ is computably random.
					\item{} ${\omega}$ is e.c.u.\ random and Schnorr random.
					\item{} ${\omega}$ is Doob random and Schnorr random.
					\item{} ${\omega}$ is Schnorr random.
				\end{enumerate}
				We have the positive results
				(1) $\Rightarrow$ (2) $\Rightarrow$ (3) $\Rightarrow$ (4) and the negative results
				(4) $\not\Rightarrow$ (3), and (2) $\not\Rightarrow$ (1).\footnote{%
					We do not know whether (3) $\Leftrightarrow$ (2), but we give a partial result in Section~\ref{sec:proof}.
				}
			\end{thm}
			\begin{proof}
				(1) $\Rightarrow$ (2) $\Rightarrow$ (3) $\Rightarrow$ (4):
					This follows from Theorem~\ref{thm:crDoob}.

				(2) $\not\Rightarrow$ (1):
					Let ${\omega}$ be such that $\omega_0$ is Schnorr random, but not computably random (such a sequence exists by a result of Wang \cite{Wang})
					and such that $\omega_{\geq 1}$ is computably random uniformly relative to $\omega_0$
					(there are measure one many such possibilities for $\omega_{\geq 1}$).
					Then ${\omega}$ is e.c.u.\ random.
					Since $\omega_0$ is not computably random, neither is ${\omega}$.

					Furthermore, $\omega_{\geq 1}$ is Schnorr random uniformly relative to $\omega_0$, which in turn is Schnorr random.
					Recall that van Lambalgen's theorem for Schnorr randomness states that for $\alpha,\beta\in 2^\N$,
					$\alpha \oplus \beta$ is Schnorr random if and only if $\alpha$ is Schnorr random and $\beta$ is Schnorr random uniformly relative to $\alpha$
					\cite{Miyabe:2011zr, 2012arXiv1209.5478M}.
					Hence ${\omega}$ is Schnorr random.

				(4) $\not\Rightarrow$ (3):
					Let ${\omega}$ be as in (2) $\not\Rightarrow$ (1), except swap the columns and rows.
					This ${\omega}$ is still Schnorr random since Schnorr randomness is preserved by computable permutations of bits \cite{2012arXiv1203.5535R}.
					For $\xi \in 2^{\N \times \N}$, label the first column as $\xi^{(0)}=(\xi_{0,0},\, \xi_{0,1}\, \xi_{0,2}\, \dots)$.
					Then $\omega^{(0)}$ is not computably random, and so there is a martingale $(M_n)$ on $2^\N$ such that $\limsup_n M_n(\omega^{(0)}) = {\infty}$.
					Convert this to a martingale $(N_n)$ on $2^{\N \times \N}$ using $N_n(\xi) =  M_n(\xi^{(0)})$.
					Therefore $\limsup_n N_n(\omega) = {\infty}$ and ${\omega}$ is not Doob random.
			\end{proof}
			\begin{thm}
				Doob randomness is incomparable with Schnorr randomness.
			\end{thm}
			\begin{proof}
				Doob $\not\Rightarrow$ Schnorr: Assume ${\omega}$ is such that  $\omega_0 = (0,0,\dots)$ and $\omega_{\geq 1}$ is computably random.
				Then ${\omega}$ is e.c.u.\ random, and hence Doob random.
				But it is not Schnorr random.

				Schnorr $\not\Rightarrow$ Doob: This is (4) $\not\Rightarrow$ (3) of Theorem \ref{thm:SchnorrDoob}.
			\end{proof}
		\subsection{In the context of Schnorr randomness}\label{sec:proof}
			As we saw in Theorem \ref{thm:DoobCharacterization}, the real $\omega$ is Doob random if
			for each computable $f$, $\omega_{B_f}$ is uniformly computably random relative to
			$\omega_{A_f}$. For fixed $f$, we may call $\omega$ $f$-Doob random in this case, and speak of a \emph{test for Doob randomness} parametrized by $f$.

			We shall prove that for each fixed $f$, there exists an $f$-Doob random, Schnorr random $\omega$ that is not e.c.u.\ random.
			However, the stronger statement of Conjecture \ref{doesntwork} would more properly complete Theorem~\ref{thm:SchnorrDoob}:

			\begin{con}\label{doesntwork}
				There is a Doob random, Schnorr random ${\omega}$ that is not e.c.u.\ random.
			\end{con}

			We will need some lemmas.
			\begin{pro}\label{pro:SRnotCR}
				For every $\alpha \in 2^\N$, there is a $\beta \in 2^\N$ which is Schnorr random uniformly relative to $\alpha$,
				but not computably random uniformly relative to $\alpha$.
				(We can even show that $\beta$ is not Schnorr random relative to $\alpha$ in the usual sense.)
			\end{pro}
			\begin{proof}
				Relativize the construction and proof of Wang \cite{Wang}.
				Namely, his argument naturally relativizes to give
				a class of martingales ${F^\alpha}$ and a class of sequences $\xi^\alpha$ such that the following holds.
				(We use the terminology and notation in his paper.)
				\begin{enumerate}
					\item{} $F^\alpha$ is uniformly computable from the oracle $\alpha$.
					\item{} $\xi^\alpha$ does not pass the test $F^\alpha$.
					\item{} $\xi^\alpha$ passes every standard Schnorr test $(F^\alpha_e, h^\alpha_e)$ which is computable from $\alpha$.
				\end{enumerate}
				These show, respectively, that $\xi^\alpha$ is not computably random uniformly relative to $\alpha$, and that
				$\xi^\alpha$ is Schnorr random relative to $\alpha$.
				Let $\beta=\xi^\alpha$.
			\end{proof}

			\subsubsection{Schnorr randomness versions of theorems of Ku\v{c}era and Miyabe}\label{sub:Schnorr_versions}
				The following Schnorr randomness version of Ku\v{c}era's theorem \cite{MR820784}
				is due to Bienvenu and Miller \cite{MR2880269}.
				\begin{pro}[Ku\v{c}era's theorem for Schnorr randomness]\label{pro:Kucera}
					Every $\Pi^0_1$ set $C \subseteq 2^\N$, such that $\mathbb{P}(C)>0$ is computable,
					contains a tail of every Schnorr random.
				\end{pro}
				\begin{proof}[Proof.]
					Let $X$ be Schnorr random. 
					Then by \cite[Theorem 9, direction (ii)-implies-(i)]{MR2880269}, some tail of $X$ belongs to $C$,
					and we are done.
				\end{proof}
				The following is a Schnorr randomness version of Miyabe's theorem extending van Lambalgen's theorem to infinitely many reals \cite{MR2675686}.
				\begin{pro}[Miyabe's extension of van Lambalgen's theorem for Schnorr randomness]\label{pro:miyabe}
					Let $\omega \in {(2^\N)}^\N$. If $\omega_n$ is Schnorr random uniformly relative to $\omega_{<n}$ for all $n$,
					then there is $\xi \in {(2^\N)}^\N$ such that $\xi$ is Schnorr random and for each $n$, $\xi_n$ is a tail of $\omega_n$.
					(In other words, there is a $k_n$ such that $\xi_n = \alpha_{\geq k_n}$ for $\alpha = \omega_n$.)
				\end{pro}
				\begin{proof}
					Given a set ${U \subseteq {(2^\N)}^\N}$ and ${\xi_{< n}}$ in ${(2^\N)}^n$, we will use the notation
					\[
						U | \xi_{< n} = \{\omega \in {(2^\N)}^\N \mid \xi_{< n} {}^{\frown}\omega \in U\}.
					\]
					We will use angle brackets to denote singleton sequences.  For example, for ${\alpha \in 2^\N}$,
					\[
						U | \langle \alpha \rangle = \{\omega \in {(2^\N)}^\N \mid (\alpha, \omega_0, \omega_1, \dots) \in U\}.
					\]
					The set ${U | \xi_{< n}}$ can be read as ``$U$ given the initial segment ${\xi_{<n}}$'' and its measure has the convenient notation
					\[
						\mathbb{P}(U | \xi_{<n}).
					\]
					Enumerate (noneffectively) all the Schnorr tests $(U^s_n)$ on ${(2^\N)}^{\N}$.
					Recall each $U^s_n$ is effectively open uniformly in $n$ and ${\mu}(U^s_n)$ is computable uniformly in $n$ and at most $2^{-n}$.

					Fix ${\omega}$ as in the statement of the proposition.  We closely follow the construction in \cite[Theorem~13]{Bienvenu-et-al}.
					At stage $s$ we construct the rows $\xi_s \in 2^\N$ of $\xi$  (along with the effectively open helper sets $V_s, W_s$) to satisfy these requirements:
					
\begin{enumerate}[label=$\mathsf{R}_{\arabic*}^s$]
\item For the $s$th Schnorr test $(U^s_n)$, we have $V_s = U^s_n$ for some $n$ and $\mathbb{P}( V_s | {\xi_{<s}})\leq 2^{-(2s+1)}$.
\item $\mathbb{P}(W_s | {\xi_{<s+1}}) \leq 1 - 2^{-(2s+2)}$ where $W_s = \bigcup_{i < s+1} V_i$.
						\item[$\mathsf{R}_3^s$.] $\xi_s = {(\omega_s)}_{\geq k}$ for some $k$ (and is therefore Schnorr random).
					\end{enumerate}
					By $\mathsf{R}_2^s$, $\xi \notin \bigcup_i V_i$.  For if $\xi \in V_i$ for some $i$, then since $V_i$ is open,
					there is some $s$ for which all sequences extending $\xi_{<s}$ in $V_i$.  Since $s$ can be arbitrary large, assume $s>i$.  Then by $\mathsf{R}_2^s$,
					\[1 = \mathbb{P}(V_i | \xi_{<s}) \leq \mathbb{P}(W_{s-1} | \xi_{<s}) \leq 1 - 2^{-(2s+2)}.\]
					A contradiction.  Therefore, $\xi$ is Schnorr random, and $\mathsf{R}_3^s$ is the desired tail property.

					\medskip \noindent \emph{Construction at stage $s$:}

					Assume $\xi_{<s} = (\xi_0, \dots, \xi_{s-1})$ has been constructed as well as $W_{s-1} = \bigcup_{i < s} V_i$
					(for $s=0$ let $\xi_{<0}$ be the empty string and $W_{-1} = \varnothing$).
					Assume they satisfy $\mathsf{R}_1^i$, $\mathsf{R}_2^i$, and $\mathsf{R}_3^i$ for $i<s$.

					Let $(U^s_n)$ be the $s$th Schnorr test.
					For any Schnorr random $\alpha$, it follows from the proof of van Lambalgen's theorem for uniformly relative Schnorr randomness
					in \cite{2012arXiv1209.5478M} that $\sum_n \mathbb{P}(U^s_n | \langle \alpha \rangle) < {\infty}$.
					(Specifically, let $t(\alpha, \psi)=\sum_n \mathbf{1}_{U^s_n}(\langle \alpha \rangle {}^{\frown}\psi)$.
					This is a nonnegative lower semicomputable function on $2^\N \times {(2^\N)}^\N$ with a computable integral.
					Then \cite[Proof of Theorem~4.1]{2012arXiv1209.5478M} shows that for any Schnorr random $\alpha$, the function $\beta \mapsto t(\alpha,\beta)$ has a finite integral.
					For our $t$, this integral is $\sum_n \mathbb{P}(U^s_n | \langle \alpha \rangle)$.)  Similarly, $\sum_n \mathbb{P}(U^s_n | \xi_{<s}) < {\infty}$.
					Hence there is some $n$ such that $\mathbb{P}(U^s_n | \xi_{<s}) \leq 2^{-(2s+1)}$.  Let $V_s = U^s_n$, thereby satisfying requirement $\mathsf{R}_1^s$.

					Recall $W_s$ from the requirements.  Since it is a finite union of $\Sigma^0_1$ sets of computable measure,
					$W_s$ also is a $\Sigma^0_1$ set of computable measure. Define
					\[
						C = \{\alpha \in  2^\N \mid \mathbb{P}(W_s|{\xi_{<s} {}^{\frown}\langle \alpha \rangle}) \leq 1-c\}
					\]
					where the computable real $c \in {\closedbegin} 2^{-(2s+2)}, 2^{-(2s+1)} {\openend}$
					is chosen as follows so that the effectively closed set $C$ has computable measure.
					To choose $c$, use Lemma \ref{lem:avoidPoints} and the function $f(\alpha):=\mathbb{P}(W_s|{\xi_{<s} {}^{\frown}\langle \alpha \rangle})$.
					It remains to show that $f$ is $L^1$-computable.
					We do this by showing $f$ is a nonnegative lower semicomputable function with computable integral \cite{2012arXiv1209.5478M}.
					First, $f$ is lower semicomputable since the measures of $\Sigma^0_1$ sets are lower semicomputable.
					Second, we have $\int f(\alpha)\, d\alpha = \mathbb{P}(W_s|{\xi_{<s}})$.
					To compute $\mathbb{P}(W_s|{\xi_{<s}})$, we will use \cite[Corollary~6.10]{MR3123251}.  		
					Namely, the value of a nonnegative lower semicomputable function,
					in this case $g(\psi_{<s})=\mathbb{P}(W_s|{\psi_{<s}})$, with a computable integral, $\int g(\psi) \,d\psi = \mathbb{P}(W_s)$,
					is computable when evaluated at a Schnorr random point, $\psi_{<s} = \xi_{<s}$.

					By Markov's inequality,
					\[
						1-\P(C)			\leq\frac{1}{1-c}\int\P(W_s \mid {\xi_{<s} {}^{\frown} \la {\alpha} \ra}) \, d{\alpha} = \frac{1}{1-c} \P(W_s \mid {\xi_{<s}})
										=    \frac{1}{1-c}\P[ (W_{s-1} \mid {\xi_{<s}}) \cup (V_s \mid {\xi_{<s}})]
					\]
					\[
										\leq\frac{1}{1-c}\left(\P(W_{s-1} \mid {\xi_{<s}})+ \P(V_s \mid {\xi_{<s}})\right)
										<    \frac{1 - 2^{- (2 (s - 1) + 2)} + 2^{- (2s + 1)}}{1 - 2^{- (2s + 1)}} = 1.
					\]
					Therefore, ${\P(C) > 0}$ and we can apply Lemma \ref{pro:Kucera} to find some ${\alpha \in C}$ such that ${\alpha}$ is a tail of ${\omega_s}$.
					Let ${\xi_s = \alpha}$, thereby satisfying requirement ${\mathsf{R}_3^s}$.  Finally, by the definition of $C$ and the choice of $c$ we have
					\[
						{\P(W_s \mid {\xi_{<s+1}}) \leq 1 - 2^{-(2s+2)}},
					\]
					satisfying ${\mathsf{R}_2^s}$.
				\end{proof}
			\subsubsection{Schnorr randomness and $f$-Doob randomness does not imply e.c.u.\ randomness}
				\begin{thm}
					Fix a computable function $f$. There exists a Schnorr random $\omega$ which is $f$-Doob random but not e.c.u.\ random.
				\end{thm}
				\begin{proof}
					Fix $f$ and give $A_f$ and $B_f$ a computable ordering. By Propositions~\ref{pro:SRnotCR} and \ref{pro:miyabe},
					there is a Schnorr random $\xi \in 2^{\N\times\N}$ such that
					for each $n$, $\xi_n$ is Schnorr random uniformly relative to $\xi_{<n}$, but
					$\xi_n$ is not computably random uniformly relative to $\xi_{<n}$. Let $\alpha \in 2^\N$ be Martin-L{\"o}f random relative to $\xi$.
					Let ${\omega}$ be defined by
					\begin{itemize}
						\item{} ${\omega_{A_f} = \xi_{A_f}}$ and
						\item{} ${\omega_{B_f}}$ is the bits of $\alpha$ put into the positions in $B_f$ using the ordering on $B_f$.
					\end{itemize}
					Note that each $\omega_n$ is Schnorr random, but not computably random.
					Then
					\begin{itemize}
						\item{} {${\omega}$ is Schnorr random}:
							Since $\xi$ is Schnorr random, so is $\xi_{A_f}=\omega_{A_f}$.
							Since $\alpha$ is Martin-L\"of random relative to $\xi$,
								$\alpha$ is uniformly Schnorr random relative to $\xi$.
							Therefore $\omega_{B_f}$ is uniformly Schnorr random relative to $\xi_{A_f}=\omega_{A_f}$.
							By van Lambalgen's theorem for uniform Schnorr randomness \cite{2012arXiv1209.5478M}, $\omega$ is Schnorr random.
						\item{} {${\omega}$ is not e.c.u.\ random}: By Corollary~\ref{cor:vLcr},
							if $\omega_{\geq n}$ is computably random uniformly relative to $\omega_{<n}$,
							then $\omega_n$ would be computably random relative to $\omega_{<n}$ which it is not.
						\item{} {${\omega}$ is $f$-Doob random} using the characterization in Theorem~\ref{thm:DoobCharacterization}.
							By the choice of $\alpha$, we have that ${\omega_{B_f}}$ is Martin-L{\"o}f random relative to $\omega_{A_f}$.
							Then, ${\omega_{B_f}}$ is also computably random uniformly relative to ${\omega_{A_f}}$  (see \cite{Miyabe:2011zr, 2012arXiv1209.5478M}).
					\end{itemize}
				\end{proof}

	\section{Doob random paths of Brownian motion}\label{sec:Brownian}
		In this section we work with the space $\Omega = \{f \in C({\closedbegin} 0,{\infty}{\openend}) : f(0)=0\}$, where $\P$ is the Wiener measure,
		i.e.\ the probability measure of Brownian motion.
		The martingale convergence theorem for Brownian motion goes back to Doob \cite{MR0058896} and is discussed in Appendix~\ref{appendix}.

		We will use the variables $W$, $X$, $Y$ and $Z$ for elements of $\Omega$.
		Let $W_t$ denote the value of $W$ at time $t$. Similar to the previous sections, we will use $W_{\leq s}$ to denote $W$ restricted to $[0,s]$.
		Also, $W_{\geq s}$ denotes the function $t \mapsto W_{s+t} - W_s$ (hence $W_{\geq s} \in \Omega$).
		Assume $W, \varW \in \Omega$.
		Then define concatenation (at time $s$) as
		\begin{equation}\label{eqn:concat}
			{(W_{\leq s} {}^{\frown}\varW)}_t =
			\begin{cases}
				W_t & t\leq s \\
				W_s + \varW_{t -s} & t > s
			\end{cases}.
		\end{equation}
		The concatenated path $W_{\leq s} {}^{\frown}\varW$
                now is computable from $W_{\leq s}$, $s$, and $\varW$.%
		\footnote{Since $C([0,s])$ may not be a computable metric space, we say that something is computable from $W_{\leq s}$ and $s$
		if it is computable from the function $\varW\in C([0,1])$ given by $\varW_t = W_{s\cdot t}$.
		Similarly, we can say that $W_{\leq s}$ is computable from $s$ and $W$, since
		the corresponding $\varW\in C([0,1])$ is computable uniformly from $s$ and $W$.
		}
		Notice that $W_{\leq s} {}^{\frown}W_{\geq s} = W$.
		
		Our definition of conditional expectation relies on these three properties of $(\Omega,\P)$.  (The second two properties, which come from the definition of Brownian motion (Defintion~\ref{df:BM}),
		are needed to establish that the properties of Remark~\ref{rem:EnProperties} still hold.)
		\begin{itemize}
			\item{} $(\Omega,\P)$ is a computable probability space.
			\item{} $(\Omega,\P)$ is \emph{stationary}, that is the map $W \mapsto W_{\geq t}$ is $(\Omega,\P)$-measure preserving.
			\item{} $(\Omega,\P)$ has \emph{independent increments}, that is the map $W \mapsto W_{\geq t}$ is independent of the map $W \mapsto W_{\leq t}$.
		\end{itemize}

		Since $(\Omega,\P)$ is a computable probability space,
		if $f:\Omega \rightarrow \mathbb{R}$ is a computable function and
		$\|f\|_{\infty} \leq C$ then $\E(f)$ is computable uniformly from $f$ and $C$ \cite[Corollary 4.3.2]{MR2519075}.
		Conditional expectation can be defined similarly to what was done in previous sections, as follows.

		\begin{df}\label{df:Et}
			Given an integrable function $f:\Omega \rightarrow \mathbb{R}$ define $\E_t (f)$ as the function given by
			\[
				\E_t (f) (W) = \int\! g(\varW)\, d\P(\varW) = \E(g)
				\quad\text{where}\quad g(\varW)=f(W_{\leq t} {}^{\frown}\varW).
			\]
		\end{df}
		Notice, if $f$ is computable and $\|f\|_{\infty} \leq C$ then $\E_t (f)(W)$ is computable uniformly from $t$, $f$, $W_{\leq t}$, and $C$.
		(See Remark~\ref{rem:bddInCondExp} for why the bound $C$ is needed.)
		Also, if $f$ is a.e.\ computable and $\|f\|_{\infty} \leq C$, then
		$W_{\leq t} \mapsto \E_t (f)(W)$ is an a.e.\ computable function uniformly from $t$, $f$, and $C$.
		The properties of Remark~\ref{rem:EnProperties} still hold.
		\begin{rem}\label{rem:bddInCondExp}
			Unlike the situation in Section \ref{sec:4}, here $\Omega$ is not compact.
			Hence a continuous function $f$ may be unbounded.
			One can show that $\E_t(f) (W)$ is not computable from $t$, $f$ and $W$, even when $f$ is integrable.
			We shall now present a function $f:\Omega \rightarrow \mathbb{R}$ such that
			\begin{itemize}
				\item{} $f$ is computable;
				\item{} $f$ is integrable, $\E |f| < +\infty$; and
				\item{} $\E_1 (f) (W) = +\infty$ for all $W\in \Omega$ with $W_1 = 0$,
			\end{itemize}
			thereby showing $\E_1 f$ is not computable.
			Let $f(W)=g(W_1, W_2 - W_1)$ for
			\[
				g(x, y) = \exp\left(\frac{ -{(x\alpha(y))}^2 + (x^2+y^2)}{2}\right)
			\]
			where
			\[
				\alpha(y) = e^{y^2/2}.
			\]
			Let $n(x)={e^{-x^2/2}/\sqrt{2\pi}}$, the p.d.f.\ of a standard $N(0,1)$ random variable.%
			\footnote{%
				This remark is the only place in this paper where we make use of the Gaussian distribution.
			}
			Note that
			\[
				g(x, y)n(x)n(y) = \frac{1}{2\pi}\exp\left(\frac{ -{(x\alpha(y))}^2}{2}\right)
				\quad \text{and}\quad
				g(0, y)n(y) = \frac{1}{\sqrt{2\pi}},
			\]
			and for $\alpha\in \mathbb{R}$,
			\[
				\int_{-{\infty}}^{\infty} e^{-{\alpha}x^2/2} dx
				= \int_{-{\infty}}^{\infty} e^{-{(\sqrt{{\alpha}}x)}^2/2} d(\sqrt{{\alpha}}x)/\sqrt{{\alpha}} = \sqrt{2\pi}/\sqrt{{\alpha}}.
			\]
			Treating $W$ as a random variable we have that $W_1$ and $W_2 - W_1$ are independent random variables each with an $N(0,1)$ normal distribution (Definition~\ref{df:BM}).  Thus
			\[
 				\E |f (W)| = \E (g(W_1, W_2-W_1)) = \int_{-{\infty}}^{\infty} \int_{-{\infty}}^{\infty}
				g(x, y)n(x)n(y)\,dx\,dy =
			\]
			\[
 				\int_{-{\infty}}^{\infty}
				\sqrt{2\pi}e^{-y^2/2}
				\,dy = 2\pi.
			\]
			And conditioning on $W_1 = 0$, we have
			\[
				\E_1 (f) (W) = \E (g(0,W_2 - W_1)) = \int_{-{\infty}}^{\infty} g(0, y)n(y)\,dy = \int_{-{\infty}}^{\infty} \frac{1}{\sqrt{2\pi}}\,dy = +{\infty}.
			\]
		\end{rem}
		\begin{df}\label{df:cont_mart}
		Say a sequence ${(M_t)}_{t\in {\closedbegin} 0,{\infty}{\openend}}$ of real-valued functions on $\Omega$ is a \textit{martingale} if both of the following hold almost surely.
			\begin{enumerate}
			\item{} $(M_t) (W)$ is \emph{adapted}, that is $M_t (W)$ is determined from $t$ and $W_{\leq t}$.
				\item{} $\E_s (M_t) = M_s$ for $s \leq t$.
			\end{enumerate}
		\end{df}
		Say that ${(M_t)}_{t\in {\closedbegin} 0,{\infty}{\openend}}$ is a \emph{computable martingale} if, furthermore, it is \emph{computably adapted},
		that is $M_t (W)$ is uniformly computable from $t$ and $W_{\leq t}$.
		\begin{pro}\label{pro:compAdaptedCont}
			A sequence ${(M_t)}_{t\in{\closedbegin} 0,{\infty}{\openend}}$ is a computable martingale if and only if
			it is a martingale and $M_t (W)$ is computable uniformly from $W$ and $t$.
			In particular the paths $t \mapsto M_t(W)$ are computable functions relative to $W$.
		\end{pro}
		\begin{proof}
			The proof is the same as Proposition~\ref{pro:compAdapted}.
		\end{proof}
		Motivated by Proposition~\ref{pro:compAdaptedCont}, a sequence ${(M_t)}_{t\in{\closedbegin} 0,{\infty}{\openend}}$ is an \emph{a.e.\ computable martingale} if
		it is a martingale and the map from $W$ to its path function $t \mapsto M_t(W)$ is an a.e.\ computable map from $\Omega$ to $C({\closedbegin} 0,{\infty}{\openend})$.

		Let $\mathbb{S}={\{s_n\}}_{n\in \N}$ be an unbounded set given by a computable sequence of nonnegative computable reals, for example $\N$ or $\mathbb{Q}^+$.
		Define \emph{$\mathbb{S}$-indexed martingales} ${(M_t)}_{t\in\mathbb{S}}$ similarly,
		that is, restrict the martingale definition to the indices $t \in \mathbb{S}$.
		An \emph{$\mathbb{S}$-indexed computable martingale},
		(respectively, an \emph{$\mathbb{S}$-indexed a.e.\ computable martingale}) is an $\mathbb{S}$-indexed martingale such that ${(M_{s_n})}_{n\in \N}$ is
		a computable sequence of computable functions (respectively, a.e.\ computable functions).
		\begin{lem}\label{lem:NtoR}
			Consider an $\mathbb{S}$-indexed computable martingale ${(M_t)}_{t\in\mathbb{S}}$ such that there is a computable bound on the growth
			$d:\mathbb{S}\rightarrow \mathbb{R}$, that is to say,
			$d(s_n)$ is computable uniformly from $n$ and $|M_s(W)| \leq d(s)$ for all $W \in \Omega$ and $s \in \mathbb{S}$.
			Then ${(M_t)}_{t\in\mathbb{S}}$ uniquely extends to a computable martingale ${(N_t)}_{t\in{\closedbegin} 0,{\infty}{\openend}}$.
			The latter is computable uniformly from the former martingale and the function $d$.%
			\footnote{If working with $L^1$-computable functions, we would not need the bound $d$.}
		\end{lem}
		\begin{proof}
			If there is such a martingale $(N_t)$ extending $(M_s)$, then for all $s\in \mathbb{S}$ and all $0<t<s$ it must satisfy
			\[
				N_t = \E_t (\E_s (M_s)) = \E_t (M_s).
			\]
			Furthermore, this equation defines a martingale $(N_t)$, and the definition of $N_t$ is invariant under the choice of $s$.
			(For $r<t<s$, $\E_r (N_t) = \E_r (\E_t M_s) = \E_r M_s = N_r$.)

			It remains to show that the martingale $(N_t)$ is a computable martingale.
			We show $ N_t (W) $ is uniformly computable from $t$ and $W_{\leq t}$ as follows.
			Fix a Cauchy-name for $t$ and use it to effectively find some $n$ such that $s_n > t$ (where $\mathbb{S} = {\{s_n\}}_{n \in\N}$).
			Recall that $N_t(W) = \E_t M_{s_n}(W)$, and that this conditional expectation is uniformly computable from
			$t$, $W_t$, $M_{s_n}$ and the bound $d(s_n)$.
		\end{proof}
		\begin{df}\label{df:Doob-random-Brownian}
			Let $W \in \Omega$. Say that $W$ is \emph{Doob random} (respectively, \emph{$\mathbb{S}$-Doob random})
			if $(M_t (W))$ converges as $t \rightarrow {\infty}$
			for all nonnegative computable martingales ${(M_t)}_{t\in{\closedbegin} 0,{\infty}{\openend}}$ (respectively, ${(M_t)}_{t\in\mathbb{S}}$).
		\end{df}
		The definition of Schnorr randomness naturally extends to any computable probability space (see for example \cite{Hoyrup-Rojas-Gacs}).
		\begin{thm}\label{thm:DoobS}
			Let $W \in \Omega$ be Schnorr random.
			Let $\mathbb{S} = \{s_0 < s_1 < \dots \}$ be the set formed by a computable, increasing, unbounded sequence of computable reals.
			Then the following are equivalent:
			\begin{enumerate}
				\item{} $(M_s (W))$ diverges as $s \rightarrow {\infty}$ for some computable nonnegative martingale ${(M_s)}_{s\in \mathbb{S}}$
					and some computable $d:\mathbb{S} \rightarrow \mathbb{R}$
					where $M_s \leq d(s)$ for all $s\in \mathbb{S}$.
				\item{} $W$ is not $\mathbb{S}$-Doob random.
				\item{} $(M_s (W))$ diverges as $s \rightarrow {\infty}$ for some a.e.\ computable nonnegative martingale ${(M_s)}_{s\in \mathbb{S}}$.
			\end{enumerate}
		\end{thm}
		\begin{proof}
			(1) $\Rightarrow$ (2):
				Immediate.

			(2) $\Rightarrow$ (3):
				This is immediate from the fact that a computable martingale is a.e.\ computable.

			(3) $\Rightarrow$ (1):
				This is the same proof as (3) $\Rightarrow$ (1) in Theorem~\ref{thm:aeMart}.
		\end{proof}
		We need the following lemma about a.e.\ computable martingales.
		\begin{lem}\label{lem:ae_comp_mart}
			Let $(M_t)$ be an a.e.\ computable martingale.
			\begin{enumerate}
				\item{} Let $\{q_i\}$ be a computable enumeration of the rationals.
					The doubly indexed sequence of local minimums ${(\min_{t\in [q_i,q_j]} M_t(W))}_{i\in \mathbb{N}, j \in \mathbb{N}}$ is a.e.\ computable from $W$.
					The same is true for the sequence of local maximums.
				\item{} From a name for $(M_t)$, we can compute a dense sequence of reals $\{a_n\}$
					which are almost-surely not local maximums or local minimums of the paths in $(M_t)$.
					That is for all $i$ and $j$, we have
					\[
						\P \left\{W \,\middle|\, a_n = \min_{t\in [q_i,q_j]} M_t(W)\right\}
						=\P \left\{W \,\middle|\, a_n = \max_{t\in [q_i,q_j]} M_t(W)\right\}=0.
					\]
				\item{} Let $a_n$ be from part~(2). Let $\tau : \Omega \rightarrow [0,{\infty}]$ be an a.e.\ computable function.%
					\footnote{An extended real $t$ is \emph{computable in $[0,{\infty}]$} exactly if $t/(t+1)$ is computable in $[0,1]$.
					Using $[0,{\infty}]$ allows us to compute the time of an event when it happens, and else wait forever.}
					Further assume $M_{\tau (W)} (W)$ almost-surely does not equal $a_n$.
					Then let $\sigma : \Omega \rightarrow [0,{\infty}]$ be the first hitting time after $\tau(W)$, that is
					\[
						\sigma(W) = \min\{t\geq \tau(W) \mid M_t(W) = a_n \}.
					\]
					This $\sigma$ is an a.e.\ computable function.
			\end{enumerate}
		\end{lem}
		\begin{proof}
			(1): By our definition of a.e.\ computable martingale,
			the path functions $t \mapsto M_t(W)$ are almost-surely computable (as continuous functions in $C({\closedbegin} 0,{\infty}{\openend})$) from $W$.
			The minimum of a continuous function $f: \mathbb{R} \rightarrow \mathbb{R}$
			over a rational compact interval $[q_i,q_j]$ is uniformly computable from $f$ and $[q_i, q_j]$
			\cite[Cor.~6.2.5]{Weihrauch2000}.
			Hence the sequence of local minimums (and maximums) is a.e.\ computable, uniformly from $W$.

			(2) This just follows from Lemma~\ref{lem:avoidPoints}
			since the sequence of local minimums and maximums is a computable sequence of a.e.\ computable functions by part (1).
			(Note the proof of Lemma~\ref{lem:avoidPoints} holds for any computable probability space, including the Wiener measure.)

			(3) It is enough to compute $\sigma(W)$ from the path $f := t \mapsto M_t(W)$ and from the starting point $t_0 := \tau(W)$ for a.e.\ $W$.
			We assumed with measure one that $f(t_0) \neq a_n$.
			Without loss of generality, $f(t_0) > a_n$.
			Using part (1), enumerate all rational intervals $[q_i, q_j]$ containing $t_0$ such that $\min_{t\in[q_i,q_j]} f(t) > a_n$.
			Then $\tau(W)$ is the supremum of the left endpoints of these intervals.  Hence $\tau(W)$ is a.e.\ lower semicomputable.
			We show the time $\tau(W)$ is also upper semicomputable as follows.  By part (2), we may assume $a_n$ is not a local maximum of $f$.
			As before, enumerate all rational intervals $[q_i, q_j]$ such that $q_i > t_0$ and $\max_{t\in[q_i,q_j]} f(t) < a_n$.
			Then $\tau(W)$ is the infimum of the left endpoints of these intervals since $a_n$ is not a local maximum.
			(It is possible that $\tau(W)$ is the infimum of the empty set, in which case $\tau(W) = {\infty}$.)
			Hence $\tau(W)$ is a.e.\ upper semicomputable, and therefore, a.e \ computable.
		\end{proof}
		\begin{thm}\label{thm:DoobCont}
			Let $W \in \Omega$ be Schnorr random. Let $\mathbb{S}$ be an enumerable unbounded set of nonnegative computable reals.
			Then the following are equivalent:
			\begin{enumerate}
			\item{} $W$ is not Doob random.
			\item{} $(M_t (W))$ diverges as $t \rightarrow {\infty}$ for some a.e.\ computable nonnegative martingale $(M_t)$.
			\item{} $\limsup_{t\rightarrow {\infty}} M_t (W) = {\infty}$ for some a.e.\ computable nonnegative martingale $(M_t)$.
				\item{} $\liminf_{t\rightarrow {\infty}} M_t (W) = {\infty}$ for some a.e.\ computable nonnegative martingale $(M_t)$.
				\item{} $\liminf_{s\rightarrow {\infty}} M_s (W) = {\infty}$ for some a.e.\ computable nonnegative martingale ${(M_s)}_{s\in \mathbb{S}}$
				\item{} $W$ is not $\mathbb{S}$-Doob random.
			\end{enumerate}
		\end{thm}
		\begin{proof}
			(1) $\Rightarrow$ (2):
				This is immediate, since computable martingales are a.e.\ computable martingales.

			(2) $\Rightarrow$ (3):
				The proof is similar to the ``upcrossing method'' used in Theorem~\ref{thm:DoobSequences}.
				Assume $a$ and $b$ are such that they are in the set $\{a_n\}$ from part (2) of Lemma~\ref{lem:ae_comp_mart}
				and such that $\liminf_t M_t(W) \leq a < b \leq \limsup_t M_t(W)$.
				Define the upcrossing and downcrossing times as before:
				\begin{align*}
					{\sigma}^\text{up}_0 &= 0 \\
					{\sigma}^\text{down}_{k} &= \inf\{t\in {\openbegin} {\sigma}^\text{up}_{k},   {\infty}{\closedend} : M_t = b\}\\
					{\sigma}^\text{up}_{k+1} &= \inf\{t\in {\openbegin} {\sigma}^\text{down}_{k}, {\infty}{\closedend} : M_t = a\}
				\end{align*}
				These times are a.e.\ computable by part (3) of Lemma~\ref{lem:ae_comp_mart}.
				Also, since the paths of $(M_t)$ are almost-surely continuous, we have that
				$\lim_k {\sigma}^\text{down}_{k}=\lim_k {\sigma}^\text{up}_{k}={\infty}$ almost-surely.

				Now define the new martingale as follows. Let $N_0 = M_0$ and
				\[
					N_t =
					\begin{cases}
						({M}_t - {M}_{{\sigma}^\text{up}_k}) + N_{{\sigma}^\text{up}_k}
						& t \in {\openbegin} {\sigma}^\text{up}_k, {\sigma}^\text{down}_k{\closedend}\\
						N_{{\sigma}^\text{down}_k}
						& t \in {\openbegin} {\sigma}^\text{down}_k, {\sigma}^\text{up}_{k+1}{\closedend}
					\end{cases}.
				\]
				Similarly to Theorem~\ref{thm:DoobSequences}, this is a well-defined martingale.
				Moreover, the paths $t \mapsto N_t(\varW)$ are almost-surely computable from $\varW$,
				since ${\sigma}^\text{up}_k(\varW)$, ${\sigma}^\text{down}_k(\varW)$, and $t \mapsto M_t(\varW)$ are almost-surely computable from $\varW$.
				(Notice the paths are continuous, which lets us compute $N_t(\varW)$ when $t$ is on the boundary of the upcrossing/downcrossing intervals.)
				Hence, $(N_t)$ is an a.e.\ computable martingale.
				Similarly to before, we have $\limsup_{t\rightarrow {\infty}} N_t (W) = {\infty}$.

			(3) $\Rightarrow$ (4): Follow the ``savings property method'' used in Theorem~\ref{thm:DoobSequences}
				with the same adaptions to the continuous setting that we made in (2) $\Rightarrow$ (3) of this theorem.

			(4) $\Rightarrow$ (5):
				Because $\liminf_{t\rightarrow {\infty}} M_t (W) = {\infty}$, it is enough to restrict $(M_t)$ to indices from the unbounded set $\mathbb{S}$.

			(5) $\Rightarrow$ (6):
				Choose an unbounded increasing subsequence $\mathbb{S}' = \{ s_0 < s_1 < \dots \} \subseteq \mathbb{S}$.
				By Theorem \ref{thm:DoobS}, we can find a computable $\mathbb{S}'$-indexed martingale ${(N_s)}_{s\in \mathbb{S}'}$
				with some computable $d$ such that $M_s \leq d(s)$ for all $s\in \mathbb{S}'$ and such that $\liminf_{s\rightarrow {\infty}} N_s (W) = {\infty}$.
				Then extend $(N_s)$ to an $\mathbb{S}$-indexed martingale by using the proof of Lemma~\ref{lem:NtoR}.

			(6) $\Rightarrow$ (1):
				By Theorem \ref{thm:DoobS}, we may assume there is some $d$ such that $M_s \leq d(s)$ for all $s\in \mathbb{S}$.
				Then by Lemma~\ref{lem:NtoR}, we can extend ${(M_s)}_{s\in \mathbb{S}}$ to a martingale indexed on the reals.
		\end{proof}
		\subsection{The isomorphism between ${(2^\N)}^\N$ and Brownian motion}
			In addition to the Wiener measure space $(\Omega,\P)$ defined above,
			there is a Wiener measure space $(C([0,1]),\P)$ which is the measure of Brownian motion in the interval $[0,1]$.
			This subsection will use the following additional facts about Brownian motion.
			\begin{itemize}
				\item{} $(C([0,1]),\P)$ is a computable probability space.
				\item{} $(C([0,1]),\P)$ is an atomless probability space.
				\item{} Brownian paths almost-surely start at $0$, i.e. $\P\{W:W_0=0\}=1$.
				\item{} The space $(\Omega,\P)$ is isomorphic to the product space ${(C[0,1],\P)}^\N$,
				through the map $W \mapsto (W_{[0,1]}, W_{[1,2]}, \ldots)$ where $W_{[s,t]}$ is $W$ restricted to $[s,t]$.
			\end{itemize}
		
			\begin{df}[{\cite[Definition~5.3]{MR2519075}}, see also \cite{2012arXiv1203.5535R}]
				An \emph{a.e.\ computable isomorphism} $B:2^\N \rightarrow \Omega$ is an a.e.\ computable function such that
				the push-forward measure of $B$ (when $2^\N$ has the fair coin measure) is the Wiener measure,
				and such that there is an a.e.\ computable inverse map $B^{-1}:\Omega \rightarrow 2^\N $ satisfying
				\begin{equation}\label{eq:aeIsomorphism}
					B^{-1}(B(\alpha))=\alpha \ \text{a.e.}\quad\text{and}\quad B(B^{-1}(W))=W\ \text{a.e.}
				\end{equation}
				One can also replace $\Omega$ with $C([0,1])$. (Note,
				(\ref{eq:aeIsomorphism}) holds for all Schnorr random (even Kurtz random) $\alpha$ and $W$ \cite{2012arXiv1203.5535R}.)
			\end{df}

			Since $C([0,1])$ with the Wiener measure is an atomless computable probability space,
			we can apply the Carath{\'e}odory isomorphism theorem to find
			an a.e.\ computable isomorphism $\widetilde{B}:2^\N \rightarrow C([0,1])$ \cite{MR2519075, kjos-hanssen-nerode}.

			We extend $\widetilde{B}$ to an a.e.\ computable isomorphism $B:{(2^\N)}^\N \rightarrow \Omega$ as follows.
			Let us denote concatenation of functions $f\in C[0,a]$ and $g\in C[0,b]$ by ${}^\frown$, i.e.,
			\[
				f^\frown g \in C[0,a+b],
			\]
			\[
				f^\frown g(x)=f(x),\quad x\in [0,a],
			\]
			\[
				f^\frown g(x)=f(a) - g(0) +g(x-a),\quad x\in [a,b].
			\]
			We note that this makes ${}^\frown$ associative.
			Let
			\begin{align*}
				B(\omega) &= \widetilde{B}(\omega_0) {}^{\frown}\widetilde{B}(\omega_1) {}^{\frown}\dots \\
				B(\omega_{< n}) &=\widetilde{B}(\omega_0) {}^{\frown}\dots {}^{\frown}\widetilde{B}(\omega_{n-1}) = {(B(\omega))}_{\leq n}\quad(n\in \N).
			\end{align*}
			For $W \in \Omega$ let $W_{[s,t]}$ be $W$ restricted to $[s,t]$.\footnote{%
				More formally $W_{[s,t]}={(W_{\geq s})}_{\leq t}$.
				}
			Then the inverse of $B$ is the a.e.\ computable function
			\begin{align*}
				B^{-1}(W) &= \left(\widetilde{B}^{-1}(W_{[0,1]}),\,\widetilde{B}^{-1}(W_{[1,2]}),\,\dots \right) \\
				B^{-1}(W_{\leq n}) &= \left(\widetilde{B}^{-1}(W_{[0,1]}),\,\dots,\widetilde{B}^{-1}(W_{[n-1,n]})\right) = {(B^{-1}(W))}_{< n}\quad(n\in\N).
			\end{align*}

			We have already mentioned Schnorr randomness on $(\Omega,\P)$.
			For \emph{computable randomness} on $(\Omega,\P)$, we will use the definition in \cite{2012arXiv1203.5535R}.
			Say that $W$ is \emph{e.c.u.\ random} if there is an $n\in \N$ such that
			$B^{-1}(W_{\geq n})$ is computably random uniformly relative to $B^{-1}(W_{\leq n})$.
			\begin{lem}\label{lem:preservation}
				Let $W\in\Omega$ be Schnorr random.  Then $\omega = B^{-1}(W)$ is Schnorr random.
				Further $W$ is, respectively, computably random, e.c.u.\ random, or Doob random, if and only if $\omega=B^{-1}(W)$ is.
			\end{lem}
			\begin{proof}
				For Schnorr and computable randomness, this follows from the fact that
				a.e.\ computable isomorphisms preserve Schnorr and computable randomness \cite{2012arXiv1203.5535R}.
				For e.c.u.\ randomness, this follows from the definition.

				For Doob randomness, by Theorem~\ref{thm:DoobCont}, it is enough to replace Doob randomness with $\N$-Doob randomness.
				We will prove one direction.
				The other is the same.
				Assume $W$ is not $\N$-Doob random.
				Then there is a martingale ${(M_n)}_{n\in \N}$ on $\Omega$ such that $(M_n(W))$ diverges as $n \rightarrow {\infty}$.
				Define a new martingale $(N_n)$ on ${(2^\N)}^\N$ by $N_n(\xi) = M_n(B(\xi))$.
				Since $B$ is a.e.\ computable, so is $(N_n)$.
				It remains to show it is a martingale.
				$N_n(\xi)$ is a.e.\ computable from ${B(\xi)}_{\leq n} =B(\xi_{< n}) $ which is computable from $\xi_{< n}$.
				Also for $m \leq n$,
				\begin{align*} \E_m (N_n) (\xi)
					&= \int_{{ (2^\N)}^\N} N_n (\xi_{<m} {}^{\frown} \psi) d\psi\\
					&= \int_{{ (2^\N)}^\N} M_n\big(B (\xi_{<m} {}^{\frown} \psi)\big) d\psi\\
					&= \int_{{ (2^\N)}^\N} M_n\big({B (\xi)}_{\leq m} {}^{\frown} B (\psi)\big) d\psi\\
					&= \int_{\Omega} M_n\big({B (\xi)}_{\leq m} {}^{\frown} \varvarW\big) d\P (\varvarW)\\
					&= \E_m (M_n) (B(\xi)) = M_m (B(\xi)) = N_m (\xi)
				\end{align*}
				which completes the proof.
			\end{proof}
			\begin{thm}\label{thm:DoobAndSchnorr}
				For a $W \in \Omega$, consider the the following:
				\begin{enumerate}
				\item{} $W$ is computably random.
					\item{} $W$ is e.c.u.\ random and Schnorr random.
					\item{} $W$ is Doob random and Schnorr random.
					\item{} $W$ is Schnorr random.
				\end{enumerate}
				We have (1) $\Rightarrow$ (2) $\Rightarrow$ (3) $\Rightarrow$ (4) and the negative results
				(4) $\not\Rightarrow$ (3), and (2) $\not\Rightarrow$ (1).\footnote{We do not know whether (3) $\Leftrightarrow$ (2).}
			\end{thm}
			\begin{proof}
				Follows from Theorem~\ref{thm:SchnorrDoob} and Lemma~\ref{lem:preservation}.
			\end{proof}
			\begin{thm}
				Doob randomness is incomparable with Schnorr randomness.
			\end{thm}
			\begin{proof}
				Doob $\not\Rightarrow$ Schnorr: Let ${\mathbf{0}}$ be the constant zero function.
				Let $Y={\mathbf{0}}_{\leq 1} {}^{\frown} X_{\geq 1}$ where ${X \in \Omega}$ is Doob random.
				Clearly, $Y$ is not Schnorr random, but we claim it is Doob random.
				Assume not.
				Then by Theorem~\ref{thm:DoobCont} there is a computable martingale ${(M_n)}_{n \in \N, n\geq 1}$ such that ${\lim_n M_n(Y) = {\infty}}$.
				Define a new martingale $N_n (W) = M_n ({\mathbf{0}}_{{\le} 1} {}^{\frown} W_{{\ge} 1})$.  To see that this is a martingale, for $m<n$,
				\begin{align*}
					(\E_m N_n)(W)	&= \int\! N_n\left(W_{{\le} m} {}^{\frown}Z\right)\, dZ\\
									&= \int\! M_n\left(\mathbf{0}_{{\le} 1} {}^{\frown}W_{[1,m]} {}^{\frown}Z\right)\, dZ\\
									&= \int\! M_n\big({(\mathbf{0}_{{\le} 1} {}^{\frown}W_{{\ge} 1})}_{{\le} m} {}^{\frown}Z \big)\, dZ\\
									&= (\E_m M_n) (\mathbf{0}_{{\le} 1} {}^{\frown}W_{{\ge} 1})\\
									&= M_m (\mathbf{0}_{{\le} 1} {}^{\frown}W_{{\ge} 1}) = N_m (W).
				\end{align*}
				Then ${\lim_n N_n(X) = {\infty}}$ which contradicts that $X$ is Doob random.

				Schnorr $\not\Rightarrow$ Doob: This is (4) $\not\Rightarrow$ (3) of Theorem \ref{thm:DoobAndSchnorr}.
			\end{proof}
	\begin{appendix}
		\section{Conditional expectation, martingales, and Doob's theorem}\label{appendix}
			As explained in the introduction, we chose definitions of conditional expectation and martingales to fit our specific situations and to make our proofs more elementary.
			The reader however, may be curious to know the more general definition of martingale used in probability theory.
			In this section we will sketch the details.
			This section is not necessary for understanding the rest of the paper.
			For more information on probability theory and discrete time martingales, see \cite{MR1155402}.
			For Brownian motion and continuous time martingales, see \cite{MR2001996}.
			For more about the computability of filtrations and martingales, see \cite{Rute:2013uq}.
	
			\subsection{Discrete-time martingales}
				Fix a probability space $(\Omega, \mathcal{B}, \P)$.
				(For simplicity assume $\Omega$ is a Polish space and $\mathcal{B}$ is its Borel $\sigma$-algebra.)
				Given a $\sigma$-algebra $\mathcal{F} \subseteq \mathcal{B}$ and an integrable function $f$,
				define the \emph{conditional expectation} $\E(f \mid \mathcal{F})$ as
				the a.e.\ unique $\mathcal{F}$-measurable function $g$ such that $\int_A g\, d\P = \int_A f\, d\P$ for all $A \in \mathcal{F}$.
				The rough idea is that $\mathcal{F}$ encodes information that we have available to us,
				and $\E(f\mid \mc{F})$ is the average or expected value of $f$ given the information that we know.
				The following are useful properties of conditional expectation (analogous to Remark~\ref{rem:EnProperties}).
				\begin{prop}[{\cite[ch.9]{MR1155402}}]\label{prop:cond_exp_properties}
					Let $f$ and $g$ be integrable functions, and let $\mc{F},\mc{G} \subseteq \mc{B}$ be $\sigma$-algebras.
					\begin{enumerate}
						\item{} If $f$ is $\mathcal{F}$-measurable, then $\E(f \mid \mathcal{F}) = f$.
						\item{} $\E(cf + g \mid \mathcal{F}) = c\, \E (f \mid \mathcal{F}) + \E(g \mid \mathcal{F})$
						for $c\in \mathbb{R}$.
						\item{} $|\E(f \mid \mathcal{F})| \leq  \E(|f| \mid \mathcal{F})$, and therefore $\|\E(f \mid \mathcal{F})\|_{\infty} \leq \|f\|_{\infty}$.
						\item{} If $f$ is $\mathcal{F}$-measurable, then $\E(fg \mid \mathcal{F}) = f\, \E(g \mid \mathcal{F})$.
						\item{} If $\mathcal{F} \subseteq \mathcal{G}$, then $\E( \E(f \mid \mathcal{G}) \mid \mathcal{F})= \E(f \mid \mathcal{F})$.
					\end{enumerate}
				\end{prop}
	
				Often it is useful to \emph{augment} a $\sigma$-algebra by adding in all null sets (and closing under the $\sigma$-algebra operations).
				The conditional expectation of the augmented $\sigma$-algebra remains the same (up to a.e.\ equivalence).
	
				A \emph{filtration} ${(\mathcal{F}_n)}_{n\in\N}$ is a sequence of $\sigma$-algebras $\mathcal{F}_n \subseteq \mathcal{B}$ such that $\mathcal{F}_n \subseteq \mathcal{F}_{n+1}$.
				(In other words, we gain information over time $n$.)
				A \emph{martingale} ${(M_n)}_{n\in\N}$ \emph{adapted} to the filtration ${(\mc{F}_n)}_{n\in\N}$ is a sequence of integrable functions
				$M_n:\Omega \rightarrow \mathbb{R}$ such that $M_n$ is $\mathcal{F}_n$ measurable
				and $\E(M_n \mid \mathcal{F}_m) = M_m$ a.s.\ whenever $m\leq n$.
				(In other words, the expected capital in the future given the past is the same as the current capital.)
				Doob's martingale convergence theorem is as follows.
				\begin{thm}[{Doob \cite{MR0058896}\cite[Theorem 11.5]{MR1155402}}]\label{thm:doob_discrete}
					If $(M_n)$ is martingale such that $\sup_n \E|M_n|<\infty$,
					then $M_n$ converges almost surely as $n \rightarrow \infty$.
				\end{thm}
	
				Notice if $(M_n)$ is a nonnegative martingale, then $\E |M_n| = \E M_n =\E M_0$, and therefore the convergence theorem applies.
	
				Given a \emph{random variable} (a measurable function) $X:\Omega \rightarrow \mc{X}$
				(where $\mathcal{X}$ is a Polish space),
				denote $\E( f \mid X)= \E(f \mid \sigma(X))$ where $\sigma(X)$ is the smallest $\sigma$-algebra containing the sets $X^{-1} (B)$ for all Borel sets $B \subseteq \mathcal{X}$.
				The \emph{distribution} $\P_X$ of $X$ is the push-forward measure on the space $\mc{X}$ given by $\P_X(A)=\P(X^{-1}(A))$.
				Two random variables, $X$ and $Y$, are said to be \emph{independent}
				if their joint distribution $\P_{(X,Y)}$ (i.e., the distribution of the pair $(X,Y)$) is equal to the product measure $\P_X \times \P_Y$.

				We will use the following lemma to give a more explicit computation of the conditional expectation operator, arriving at the definitions used in this paper.
	
				\begin{lem}[{\cite[Theorem~6.2.1]{MR2722836}}]\label{lem:ind_cond_exp}
					Let $X$ and $Y$ be independent random variables.
					Let $\varphi$ be a function with $\E|\varphi(X,Y)| < \infty$ and let $g(x)=\E(\varphi(x,Y))$.
					Then
					\[ \E(\varphi(X,Y)\mid X) = g(X) \quad \text{a.s.}\]
					In other words, for a.e. $\omega \in \Omega$,
					\[ \E(\varphi(X,Y)\mid X)(\omega) = \int \varphi(X(\omega),y)\, d\P_Y(y).\]
				\end{lem}
	
				The main idea is that $X$ encodes the past, and $Y$ encodes the future.  An integrable random variable $f$ can be expressed as a function of the past and the future, $\varphi(X,Y)$.
				If the past $X$ and future $Y$ are independent (as is the case with repeated coin-flipping), then Lemma~\ref{lem:ind_cond_exp} gives us an explicit formula for the conditional expectation.
	
				\begin{ex}\label{ex:1}
					Let $(\Omega, \P)$ be  ${(2^\N)}^\N$ with the the uniform measure as in Section~\ref{sec:3}.
					Let $Z_n: {(2^\N)}^\N \rightarrow 2^\N$ be $Z_n(\omega) = \omega_n$.
					Let $X_n = (Z_0, \ldots, Z_{n-1})$ and $Y_n = (Z_{n},Z_{n+1},\ldots)$.
					Fix $n$.  Then $X_n$ and $Y_n$ are independent, and $\P_{Y_n} = \P$.
					Given an integrable $f$, let $\varphi(x,y) = f(x {}^\frown y)$.
					Combining $f = \varphi(X_n,Y_n)$, Lemma~\ref{lem:ind_cond_exp}, and Definition~\ref{df:cond_exp_seq_seq}, we have
					\[
						\E(f\mid X_n)
					 	= \E(\varphi(X_n,Y_n)\mid X_n)
					 	= \int \varphi(X_n,\beta)d\P(\beta)
					 	= \E_n(f) \quad \text{a.s.}
					\]
					Let $\mc{F}_n$ be the augmentation of $\sigma(X_n)$.
					Then ${(\mc{F})}_{n\in\N}$ is a filtration since $\sigma(Z_0,\ldots, Z_{n-1}) \subseteq \sigma(Z_0,\ldots, Z_{n-1}, Z_{n})$.
					An integrable function $f$ is $\mc{F}_n$-measurable if and only if for a.e.\ $\alpha$ the value of $f(\alpha)$ depends only on $X_n(\alpha)=\alpha_{<n}$.
					Therefore, $(M_n)$ is a martingale in the sense of Definition~\ref{df:mart_seq_seq} exactly
					if it is a martingale adapted to the filtration $\mc{F}_n$.
		
					A similar result holds for the martingales on $2^\N$ that were introduced in Section~\ref{sec:2}.
				\end{ex}

				More generally, given a partial order $(J,\leq)$, a \emph{filtration} ${(\mathcal{F}_j)}_{j\in J}$ is
				a $J$-indexed family of  $\sigma$-algebras $\mathcal{F}_j \subseteq \mathcal{B}$ such that $\mathcal{F}_i \subseteq  \mathcal{F}_{j}$ whenever $i \leq j$.
				A \emph{martingale} ${(M_j)}_{j\in J}$ \emph{adapted} to the filtration ${(\mathcal{F}_j)}_{j\in J}$ is a sequence of integrable functions
				$M_j:\Omega \rightarrow \mathbb{R}$ such that $M_j$ is $\mathcal{F}_j$ measurable
				and $\E(M_j \mid \mathcal{F}_i) = M_i$ a.s.\ whenever $i\leq j$.

				\begin{thm}[Doob, see {\cite[Ch.II Theorem~2.5 and Corollary~2.4]{MR1725357}}]\label{Doob_conv_linear}
					Assume $(J,\leq)$ is countable and linearly ordered.
					(For notational convenience assume $J \subseteq [0,1]$ and $\sup J = 1$.)
					Assume ${(M_j)}_{j\in J}$ is a martingale such that $\sup_{j \in J} \E|M_j|<\infty$.
					Then $M_j$ converges almost surely as $j \rightarrow 1$.
					Furthermore, if $1 \in J$ and $\mc{F}_1 = \sigma(\bigcup_{j<1} \mc{F}_i)$,
					then $\lim_{j \rightarrow 1} M_j = M_1$ (both a.e.\ and in the $L^1$-norm).
				\end{thm}
	
				\begin{ex}\label{ex:2}
					Let $(\Omega,\P)$ be $2^{\N \times \N}$ with the uniform measure $\P$ as in Section~\ref{sec:4}.
					Let $(J,\leq)$ be $\N \times \N$ with the lexiographic order.
					Let $Z_{m,n}: 2^{\N\times\N} \rightarrow \{0,1\}$ be $Z_{m,n}(\omega) = \omega_{m,n}$, and
					let $X_{m,n} = (Z_{i,j} \mid (i,j) < (m,n))$.
					By the same argument as the previous example, we have $\mathbb{E}_{m,n}(f) = \E(f\mid X_{m,n})$ a.s.
					(where $\E_{m,n}$ is as in Definition~\ref{df:cond_exp_mart_array}).
					Let $\mathcal{F}_{m,n}$ be the augmentation of $\sigma(X_{m,n})$.
					The martingales in Definition~\ref{df:cond_exp_mart_array} are exactly the martingales adapted to the filtration
					${(\mathcal{F}_{m,n})}_{(m,n)\in \N \times \N}$.
		
					Notice the filtration ${(\mathcal{F}_m)}_{m\in\N}$ from Example~\ref{ex:1}
					is equal to the subfiltration ${(\mathcal{F}_{m,0})}_{m\in\N}$
					(under the natural identification of ${(2^\N)}^\N$ and $2^{\N\times\N}$).
					Furthermore, $\mathcal{F}_{m+1,0}=\sigma(\bigcup_n \mc{F}_{m,n})$.
					Therefore, $M_{m+1,0} = \lim_n M_{m,n}$ by Theorem~\ref{Doob_conv_linear}.
		
					Now, for each $D \subseteq \N \times \N$, let $X_D = (Z_{i,j} \mid (i,j) \in D)$.
					We have $\mathbb{E}_D(f) = \E(f\mid X_D)$ a.s.\ (where $\E_D$ is as in Definition~\ref{df:cond_exp_set_index}).
				\end{ex}
			\subsection{Continuous-time martingales and Brownian motion.}
				In the continuous-time setting,
				using the index set $J=\closedbegin 0,\infty \openend$,
				the definition of filtration and martingale are the same as above.
				However, for Doob's martingale convergence theorem to hold,
				we also require that ${(M_t)}_{t\in\closedbegin 0,\infty \openend}$ be \emph{a.s.\ right continuous},
				that is for almost every $\omega \in \Omega$,
				if $s\in  \closedbegin 0,\infty \openend$ then $\lim_{t \rightarrow s^+} M_t(\omega) = M_{s}(\omega)$.
				\begin{thm}[Doob {\cite{MR0058896}\cite[Theorem C.5]{MR2001996}}]\label{thm:doob_cont}
					If ${(M_t)}_{t\in\closedbegin 0,\infty \openend}$ is an a.s.\ right-continuous martingale such that $\sup_t \E|M_t|<\infty$,
					then $(M_t)$ converges almost surely as $t \rightarrow \infty$.
				\end{thm}
				\begin{df}\label{df:BM}
				\emph{Brownian motion} is defined as a sequence of random variables ${(B_t)}_t\in\closedbegin 0,\infty\openend$ with the following properties.
				\begin{enumerate}
					\item{} $B_0 = 0$ almost surely.
					\item{} $B_t-B_s$ has a normal distribution with mean $0$ and variance $t-s$ independent of $\{B_u\colon u\le s\}$, for all $t>s\ge 0$.
					\item{} $B$ has continuous paths, $t \mapsto B_t$, almost surely.	
				\end{enumerate}
				\end{df}
				Since the paths are continuous, Brownian motion can be characterized by the distribution of its continuous paths.  This is the \emph{Wiener measure} on $C(\closedbegin 0,\infty \openend)$ or $C([0,1])$.

				Let $(\Omega,\P)$ be $C(\closedbegin 0,\infty \openend)$ with the Wiener measure.
				As in Section~\ref{sec:Brownian}, let $W$ denote an element of $\Omega$.  Let $X_{t}(W) = W_{\leq t}$ and let $\mc{F}_t$ be the augmentation of $\sigma(X_{t})$.
				(Because of continuity, $\sigma(X_{t})$ is also equal to the smallest $\sigma$-algebra generated by all sets of the form
				$\{f \in \Omega : f(t_1) \in B_1, \ldots, f(t_k) \in B_k\}$ for $0 \leq t_1 < \ldots < t_k \leq t$ and Borel $B_1,\ldots,B_k \subseteq \mathbb{R}$.)
				This $(\mathcal{F}_t)$ is know as the \emph{filtration of Brownian motion} and the martingales adapted to $(\mathcal{F}_t)$ are known as \emph{Brownian martingales}.
		
				By the same proof as in Example~\ref{ex:1} we have for each $0 \leq t < \infty $ that $\E_t f = \E(f \mid \mc{F}_t)$ a.s.\ where $\E_t$ is as in Definition~\ref{df:Et}.
				We can see that an integrable function $f$ is $\mc{F}_t$-measurable if and only if $f(W)$ is almost surely determined by $t$ and $W_{\leq t}$.
				Therefore, the martingales in Definition~\ref{df:cont_mart} are exactly the Brownian martingales.
		
				The computable and a.e.\ computable martingales of Section~\ref{sec:Brownian} were defined so that the paths $t \mapsto M_t(W)$ are a.s.\ continuous (Proposition~\ref{pro:compAdaptedCont} and the following paragraph).
				Therefore, Doob's martingale convergence theorem applies.
				Conversely, it is a theorem that every Brownian martingale has a version with continuous paths \cite[Ch.V Theorem~3.2]{MR1725357}.
				(A martingale $(M_t)$ is said to be a \emph{version} or \emph{modification} of another martingale $(N_t)$ if for all $t$, $\P\{W : M_t(W) = N_t(W)\}=1$.)
	\end{appendix}

\section*{Acknowledgments}
We would like to thank the anonymous referees for their helpful feedback.

	\bibliographystyle{plain}
	\bibliography{KNR}
\end{document}